\documentclass[sigconf, nonacm]{acmart}
\usepackage{amsmath,amsfonts}
\usepackage{algorithm}
\usepackage[noend]{algpseudocode}

\usepackage{graphicx}
\usepackage{textcomp}
\usepackage{xcolor}

\usepackage{amsthm}
\usepackage{mathtools}
\usepackage{dsfont}

\usepackage{tikz}
\usetikzlibrary{calc}
\usepackage{zref-savepos}

\usepackage{url}
\usepackage{subcaption}

\usepackage{multirow}
\usepackage{setspace}
\usepackage{footnote}

\def\tg{\widetilde{g}}
\def\tm{\widetilde{m}}
\def\tv{\widetilde{v}}

\def\R{{\mathbb R}}

\def\mG{\mathcal{G}}

\def\E{{\mathbb{E}}}
\def\1{{\mathds{1}}}
\def\P{{\mathbb{P}}}

\newcommand{\ip}[2]{\left\langle #1, #2 \right \rangle}

\newtheorem{theorem}{Theorem}
\newtheorem{lemma}{Lemma}
\newtheorem{corollary}{Corollary}
\newtheorem{assumption}{Assumption}
\newtheorem{definition}{Definition}

\renewcommand{\comment}[1]{%
 \tag*{$\triangleright$ #1}
}

\newcommand{\mcir}[1]{{\mbox{\large \textcircled{\small #1}}}}

\newcounter{NoTableEntry}
\renewcommand*{\theNoTableEntry}{NTE-\the\value{NoTableEntry}}
\newcommand*{\strike}[2]{%
  \multicolumn{1}{#1}{%
    \stepcounter{NoTableEntry}%
    \vadjust pre{\zsavepos{\theNoTableEntry t}}
    \vadjust{\zsavepos{\theNoTableEntry b}}
    \zsavepos{\theNoTableEntry l}
    \hspace{0pt plus 1filll}%
    #2
    \hspace{0pt plus 1filll}%
    \zsavepos{\theNoTableEntry r}
    \tikz[overlay]{%
      \draw
        let
          \n{llx}={\zposx{\theNoTableEntry l}sp-\zposx{\theNoTableEntry r}sp-\tabcolsep},
          \n{urx}={\tabcolsep},
          \n{lly}={\zposy{\theNoTableEntry b}sp-\zposy{\theNoTableEntry r}sp},
          \n{ury}={\zposy{\theNoTableEntry t}sp-\zposy{\theNoTableEntry r}sp}
        in
        (\n{llx}, \n{lly}) -- (\n{urx}, \n{ury})
      ;
    }%
  }%
}

\usepackage{xfakebold}
\newcommand{\fbseries}{\unskip\setBold\aftergroup\unsetBold\aftergroup\ignorespaces}
\makeatletter
\newcommand{\setBoldness}[1]{\def\fake@bold{#1}}
\makeatother
\setBoldness{1}
\newsavebox\CBox
\def\textsb#1{\sbox\CBox{#1}\resizebox{\wd\CBox}{\ht\CBox}{{\fbseries{#1}}}}

\AtBeginDocument{%
  \providecommand\BibTeX{{%
    \normalfont B\kern-0.5em{\scshape i\kern-0.25em b}\kern-0.8em\TeX}}}

\begin{document}

\setlength{\abovedisplayskip}{2pt}
\setlength{\belowdisplayskip}{2pt}
\setlength{\abovedisplayshortskip}{1pt}
\setlength{\belowdisplayshortskip}{1pt}

\title{Compressed Communication for Distributed Training: Adaptive Methods and System}


\author{Yuchen Zhong}
\email{yczhong@cs.hku.hk}
\affiliation{%
 \institution{The University of Hong Kong}
 \country{}
}

\author{Cong Xie}
\email{cx2@illinois.edu}
\affiliation{%
 \institution{University of Illinois Urbana-Champaign}
 \country{}
}

\author{Shuai Zheng}
\email{shzheng@amazon.com}
\affiliation{%
 \institution{Amazon Web Services}
 \country{}
}

\author{Haibin Lin}
\email{haibin.lin.aws@gmail.com}
\affiliation{%
 \institution{ByteDance Inc.}
 \country{}
}


\renewcommand{\shortauthors}{}

\begin{abstract}
  Communication overhead severely hinders the scalability of distributed machine learning systems. Recently, there has been a growing interest in using gradient compression to reduce the communication overhead of the distributed training. However, there is little understanding of applying gradient compression to adaptive gradient methods. Moreover, its performance benefits are often limited by the non-negligible compression overhead. In this paper, we first introduce a novel adaptive gradient method with gradient compression. We show that the proposed method has a convergence rate of $\mathcal{O}(1/\sqrt{T})$ for non-convex problems. In addition, we develop a scalable system called BytePS-Compress for two-way compression, where the gradients are compressed in both directions between workers and parameter servers. BytePS-Compress pipelines the compression and decompression on CPUs and achieves a high degree of parallelism. Empirical evaluations show that we improve the training time of ResNet50, VGG16, and BERT-base by 5.0\%, 58.1\%, 23.3\%, respectively, without any accuracy loss with 25 Gb/s networking. Furthermore, for training the BERT models, we achieve a compression rate of 333x compared to the mixed-precision training.
\end{abstract}



\maketitle

\section{Introduction}

Large-scale machine learning is becoming increasingly important. Recent results suggest that larger models and datasets yield a substantial gain in model performance \cite{ kaplan2020scaling}. For example, BERT \cite{devlin2018bert}, a large bidirectional self-attention model involving 336M parameters, is trained on a large Wikipedia corpus and has reached state-of-the-art in many benchmarks. Likewise, GPT-3 \cite{brown2020language} uses 175B parameters to achieve excellent results in natural language generation. However, training such a large model is far exceeding the computing capability provided by a single machine. Thus, it is necessary to distribute the computation on multiple computing nodes. 

The most widely used paradigm for parallelizing DNN training is data parallelism. In data parallelism, each worker computes gradients on a partition of data, and then the local gradients are aggregated and broadcast to all workers. Unfortunately, the scalability of data parallelism is severely hindered by the high communication overhead. It has been observed that such overhead can take up 50\% to 90\% of the total time for representative DNN models even when equipped with high-speed interconnects such as NVLink \cite{narayanan2019pipedream}. 

One idea is to perform lossy compression of the gradients, such as quantization \cite{seide20141, alistarh2017qsgd,  bernstein2018signsgd, wen2017terngrad, zheng2019communication}, and sparsification \cite{aji2017sparse, stich2018sparsified, wangni2018gradient} to mitigate the communication bottleneck. It has demonstrated remarkable success in large neural network training on many machine learning tasks, such as image classification on ImageNet \cite{alistarh2017qsgd, bernstein2018signsgd, basu2019qsparse,  zheng2019communication, xie2020cser}, neural machine translation \cite{aji2017sparse}, and speech recognition \cite{seide20141,lin2017deep}. 

However, prior works mainly focus on distributed stochastic gradient descent (SGD) algorithm and its variants like SGD with momentum (SGDM)~\cite{alistarh2017qsgd, bernstein2018signsgd, wen2017terngrad, basu2019qsparse, stich2018sparsified, zheng2019communication, xie2020cser}. There is little attention towards applying gradient compression to adaptive gradient methods like Adam \cite{kingma2014adam}, which has been observed to outperform SGD on some tasks, such as BERT pretraining \cite{zhang2020adaptive}. Adaptive gradient methods dynamically adjust the coordinate-wise learning rate and have a nonlinear dependency on the gradient, which leads to a biased estimator and complicates the theoretical analysis. 

Moreover, despite gradient compression significantly reducing the communication cost, compression overhead is non-negligible and can impede the scalability, slowing down the training in practice \cite{xu2020compressed, m2021efficient}.  Recently, \citeauthor{xu2020compressed} (\citeyear{xu2020compressed}) developed a unified framework called GRACE which supports various gradient compression methods~\cite{xu2020compressed}. However, it is shown that only random-$k$ (random $k$ elements of the gradient are sent) obtains higher throughput than its counterpart that does not use gradient compression for training ResNet50 \cite{he2016identity}. In contrast, all the other compression methods perform significantly worse \cite{xu2020compressed}. The reason is that the compression overhead exceeds the saved communication time, making it worse than not using compression at all.

Given the above two challenges, we ask, \emph{How can we develop a fast adaptive gradient method with gradient compression and design a better system to scale up various sophisticated compression methods?} This paper answers the question. First, we introduce a novel adaptive gradient method with gradient compression that can be used to train Transformer-based models such as BERT. Second, we present a scalable system that addresses the practical issue of improving the training time by leveraging state-of-the-art gradient compression methods. The proposed system adopts parameter-server architecture and considers two-way compression, in which gradients are compressed when they are pushed to and pulled from workers. Our system is based on BytePS \cite{jiang2020unified}, an open-source state-of-the-art parameter-server implementation.

Our contributions are three-fold: 
\begin{enumerate}
    \item We propose a novel adaptive gradient algorithm called compressed LANS (CLAN). We show that CLAN achieves the same convergence rate as its full-precision counterpart.
    \item We design a scalable system called BytePS-Compress for gradient compression. BytePS-Compress allows two-way compression and is highly optimized, achieving better speedup than the full-precision methods.
    \item Empirical evaluations show that BytePS-Compress is efficient and scalable in training different deep neural networks. We improve the training time of ResNet50, VGG16, and BERT-base using CLAN by 5.0\%, 58.1\%, 23.3\%, respectively, without any accuracy loss. For training the BERT models, we obtain a compression rate of 333x. 
\end{enumerate}

\section{Related Work}

\subsection{Communication Architecture}
Parameter-server \cite{dean2012large, li2014scaling} and All-Reduce are two commonly adopted architectures for distributed training. 
Algorithm~\ref{alg:ga} shows how the local gradients are aggregated through the parameter-server architecture with $n$ workers. Each worker sends its local gradient $g_{t, i}$ to the server, and the server averages the gradients and sends it back to each worker. 
On the other hand, All-Reduce computes the sum of gradients across devices, and each worker updates its local parameters accordingly. Compared to parameter-server, All-Reduce incurs more synchronization cost but enjoys a low communication cost that is nearly independent of the number of workers \cite{peng2019generic}. Recent results show that All-Reduce has achieved almost linear speedup in large batch training on ImageNet \cite{goyal2017accurate}, while parameter-server can obtain twice the bandwidth advantage on cheap CPU nodes \cite{peng2019generic}. A recent advance on distributed framework called BytePS combines these two architectures for better performance \cite{jiang2020unified} to get the benefits from both. Specifically, BytePS uses All-Reduce for intra-node communication and parameter-server for inter-node communication. 

\vspace*{-.05in}
 \begin{algorithm}
\caption{Gradient Aggregation} \label{alg:ga}
\begin{algorithmic}[1]
\State \textbf{Input:} $g_{t, i}$

\State \textbf{function} push\_pull($g_{t, i}$)

    \State \hspace*{\algorithmicindent} \textbf{On worker i}

    \State \hspace*{\algorithmicindent} \hspace*{\algorithmicindent} \textbf{Push} $g_{t, i}$ \textbf{ to server}

    \State \hspace*{\algorithmicindent} \textbf{On server}

    \State \hspace*{\algorithmicindent} \hspace*{\algorithmicindent} \textbf{Pull} $g_{t, i}$ \textbf{ from each worker}

    \State \hspace*{\algorithmicindent} \hspace*{\algorithmicindent} \textbf{Push} $p_t = \frac{1}{n}\sum_{i=1}^n g_{t, i}$ \textbf{ to each worker}
        
\State \textbf{end function}
\end{algorithmic}
\end{algorithm}
\vspace*{-.25in}

\subsection{Adaptive Gradient Methods} 
Adaptive methods such as Adam \cite{kingma2014adam} have shown excellent results in many deep learning tasks, such as training the BERT model \cite{devlin2018bert}. Adaptive methods dynamically adjust the learning rate for each parameter with an exponential average of the second-order moment of the gradients. 
For example, \citeauthor{you2019large} (\citeyear{you2019large}) developed a layer-wise adaptive method called LAMB~\cite{you2019large}. It is shown that LAMB reduces the BERT pretraining time from 3 days to 76 minutes on 1024 TPUv3 chips. Furthermore, a new accelerated gradient method called LANS is recently introduced \cite{zheng2020accelerated}, incorporating Nesterov's momentum into LAMB. It further reduces BERT pretraining time to 54 minutes on 1536 GPUs without any performance deterioration. The LANS method without gradient normalization under the parameter-server architecture is shown in Algorithm~\ref{alg:lans}. In the algorithm, we partition the gradient into $B$ blocks, where each block $b$ has $d_b$ elements indexed by $\mathcal{G}_b$. 

\vspace*{-.1in}
\begin{algorithm}
    \caption{LANS \cite{zheng2020accelerated}}\label{alg:lans}
\begin{algorithmic}[1]
\State \textbf{Input:} $x_1\in R^d$, learning rate $\{\eta_t\}^T_{t=1}$, parameters $0 < \beta_1, \beta_2 < 1$, scaling function $\phi$, $\epsilon > 0$

\State \textbf{Initialize:} $m_0 = 0$, $v_0 = 0$

\State \textbf{for} $t=1$ \textbf{to} $T$ \textbf{do} 
    \State \hspace*{\algorithmicindent} {\bf On each worker i}
    \State \hspace*{\algorithmicindent} \hspace*{\algorithmicindent} Compute mini-batch stochastic gradient $g_{t, i}$
    \State \hspace*{\algorithmicindent} \hspace*{\algorithmicindent} $\tg_{t}$ = push\_pull($g_{t, i}$)
    \State \hspace*{\algorithmicindent} \hspace*{\algorithmicindent} \textbf{for} $b=1$ \textbf{to} $B$ \textbf{do} 
    
        \State \hspace*{\algorithmicindent} \hspace*{\algorithmicindent} \hspace*{\algorithmicindent} $m_{t, \mG_b} = \beta_1m_{t-1, \mG_b} + (1 - \beta_1)\tg_{t, \mG_b}$
        \State \hspace*{\algorithmicindent} \hspace*{\algorithmicindent} \hspace*{\algorithmicindent} $v_{t, \mG_b} = \beta_2v_{t-1, \mG_b} + (1 - \beta_2)\tg_{t, \mG_b}^2$
        \State \hspace*{\algorithmicindent} \hspace*{\algorithmicindent}  \hspace*{\algorithmicindent} $\tm_{t, \mG_b} = m_{t, \mG_b} / (1 - \beta_1^t)$
        \State \hspace*{\algorithmicindent} \hspace*{\algorithmicindent} \hspace*{\algorithmicindent} $\tv_{t, \mG_b} = v_{t, \mG_b} / (1 - \beta_2^t)$
        \State \hspace*{\algorithmicindent} \hspace*{\algorithmicindent} \hspace*{\algorithmicindent} compute ratios $r_{t, \mG_b} = \frac{\tm_{t, \mG_b}}{\sqrt{\tv_{t, \mG_b}} + \epsilon}$, $c_{t, \mG_b} = \frac{\tg_{t, \mG_b}}{\sqrt{\tv_{t, \mG_b}} + \epsilon}$
        \State \hspace*{\algorithmicindent} \hspace*{\algorithmicindent} \hspace*{\algorithmicindent} 
        $\begin{aligned}[t]
        \tilde{d}_{t, \mG_b} & = \phi(\|x_{t, \mG_b}\|_2)\left[\frac{\beta_1}{\|r_{t, \mG_b} + \lambda x_{t, \mG_b}\|}(r_{t, \mG_b} + \lambda x_{t, \mG_b}) \right. \\
        & \left. + \frac{1 - \beta_1}{\|c_{t, \mG_b} + \lambda x_{t, \mG_b}\|}(c_{t, \mG_b} + \lambda x_{t, \mG_b})\right]
        \end{aligned}$
        \State \hspace*{\algorithmicindent} \hspace*{\algorithmicindent} \hspace*{\algorithmicindent} $x_{t + 1, \mG_b} = x_{t, \mG_b} - \eta_t\tilde{d}_{t, \mG_b}$
        
    \State \hspace*{\algorithmicindent}  \hspace*{\algorithmicindent} \textbf{end for}
\State \textbf{end for}
\end{algorithmic}
\end{algorithm}
\vspace*{-.2in}

\subsection{Gradient Compression} 
Recently there has been a surge of interest in using gradient compression to ease the communication bottleneck. A few notable compression methods are designed to reduce communication costs: 1) 1-bit SGD with error feedback is introduced to accelerate the training of the speech recognition system  \cite{seide20141}; 2) QSGD \cite{alistarh2017qsgd} uses stochastic rounding to ensure that the compressed gradient is unbiased; 
3) dist-EF-SGD \cite{zheng2019communication} uses block-wise scaled 1-bit for two-way compression and has achieved state-of-the-art results. However, nearly all current works focus on SGD and its variants like momentum SGD, leaving adaptive gradient methods such as Adam behind. Very recently, \citeauthor{tang2020apmsqueeze} (\citeyear{tang2020apmsqueeze}) developed an Adam-like preconditioned method with gradient compression~\cite{tang2020apmsqueeze}. However, it requires a full-precision warm-up, which incurs extra hyper-parameter tuning. Furthermore, when gradient compression is enabled, the second-order momentum $v_t$ is fixed as the constant. Thus, this does not maintain the nature of adaptive gradient methods that the coordinate-wise learning rate is kept updating on the fly. \citeauthor{chen2020efficient} (\citeyear{chen2020efficient}) proposed to aggregate the compressed update of Adam \cite{chen2020efficient}. However, it does not achieve speedup using more workers and larger batch size. In this paper, we introduce compressed LANS (CLAN) that attains the same convergence rate as its full-precision counterpart, and the convergence is improved when we increase the number of workers and batch size. Moreover, the compression overhead is non-negligible, and inefficient system design and implementation can slow the training \cite{xu2020compressed}. Thus, we need a scalable system that can bridge the gap between theory and practice.

\section{Adaptive Gradient Methods with Gradient Compression}

In order to develop a fast adaptive gradient method with gradient compression, there are several challenges that we need to overcome: 1) enabling error feedback for the adaptive gradient method; 2) ensuring that the training is scaled by increasing the number of nodes and batch size; 3) theoretically proving that the proposed method converges for non-convex problems. 

\subsection{Error Feedback}

The standard way of applying gradient compression to SGD can be formulated as follows:
\begin{align*}
x_{t+1} \quad=\quad x_t - \eta\mathcal{C}(g_t), 
\end{align*}
where $\mathcal{C}$ is the compressor that is used to compress the gradient into a low-bit format such that the communication cost can be reduced. Such a straightforward strategy works well with unbiased compressors such as random-$k$ \citep{horvath2020better} but may diverge with biased compressors such as 1-bit \citep{karimireddy2019error}. Recently, error feedback \citep{karimireddy2019error} is introduced to circumvent the divergence issue. The idea is to store the compression error $e_t$ and use it to correct the update direction in the next iteration:
\begin{alignat*}{2}
x_{t+1} & \quad=\quad && x_t - \eta\mathcal{C}(g_t + e_t), \\
e_{t+1} & \quad=\quad && g_t + e_t - \mathcal{C}(g_t + e_t).
\end{alignat*}
It is shown that such recursion converges with arbitrary compressors that satisfy certain conditions. \citeauthor{zheng2019communication} (\citeyear{zheng2019communication}) further shows how we can apply error feedback to the SGD method for the distributed training under the parameter-server architecture~\cite{zheng2019communication}:
\begin{alignat*}{2}
x_{t+1} &\quad = \quad&& x_t - \eta\mathcal{C}\left(\frac{1}{n}\sum_{i=1}^n\mathcal{C}(g_{t, i} + e_{t, i}) + \tilde{e}_t\right), \\
e_{t+1, i} &\quad = \quad&& g_{t, i} + e_{t, i} - \mathcal{C}(g_{t, i} + e_{t, i}), \\
\tilde{e}_{t+1} &\quad = \quad&& \frac{1}{n}\sum_{i=1}^n\mathcal{C}(g_{t, i} + e_{t, i}) + \tilde{e}_t - \mathcal{C}\left(\frac{1}{n}\sum_{i=1}^n\mathcal{C}(g_{t, i} + e_{t, i}) + \tilde{e}_t\right),
\end{alignat*}
where $e_{t, i}$ is the local error stored on the worker $i$ and $\tilde{e}_t$ is the compression error kept on the server. The method recovers its full-precision counterpart when $\mathcal{C}$ is an identity mapping, i.e., $e_{t, i} = \tilde{e} = 0$ for all $t$ and $i$. While SGD methods have a linear dependency on the gradient, this is not the case for adaptive methods. Steps~12 and ~13 in Algorithm~\ref{alg:lans} show that the update $\tilde{d}_t$ is nonlinearly dependent on the gradient. This nonlinear dependency results in a biased estimator and makes it difficult to analyze the convergence in expectation, which complicates the case where gradient compression is adopted.
%
\subsection{Compressed LANS}

This paper presents a novel method called compressed LANS (CLAN) with convergence guarantees and scales with more workers and larger batch size. CLAN replaces {\bf push\_pull} in Algorithm~\ref{alg:lans} with its compressed version. We develop two types of compressed gradient synchronization for different types of compressors: 1) {\bf compress\_push\_pull}, a simple way of compressing gradient in both directions for unbiased compressors, as a recent study shows that unbiased compressors have better convergence without error feedback \citep{horvath2020better}; 2) {\bf compress\_ef\_push\_pull}, a two-way gradient compression with error feedback for biased compressors.

\begin{algorithm}
\caption{Gradient Aggregation with Gradient Compression}
\label{alg:gagc}
\begin{algorithmic}[1]
\State \textbf{Input:} compressor $\mathcal{C}$, gradient $g_{t, i}$

\State \textbf{function} compress\_push\_pull($\mathcal{C}$, $g_{t, i}$)

    \State \hspace*{\algorithmicindent} \textbf{On worker i}

    \State \hspace*{\algorithmicindent} \hspace*{\algorithmicindent} \textbf{Push} $\delta_{t, i} = \mathcal{C}(g_{t, i})$ \textbf{ to server}

    \State \hspace*{\algorithmicindent} \textbf{On server}

    \State \hspace*{\algorithmicindent} \hspace*{\algorithmicindent} \textbf{Pull} $\delta_{t, i}$ \textbf{ from each worker}

    \State \hspace*{\algorithmicindent}  \hspace*{\algorithmicindent} $\Delta_t = \frac{1}{n}\sum_{i=1}^n\delta_{t, i}$
        
    \State \hspace*{\algorithmicindent}  \hspace*{\algorithmicindent} $p_t = \mathcal{C}(\Delta_t)$

    \State \hspace*{\algorithmicindent} \hspace*{\algorithmicindent} \textbf{Push} $p_t$ \textbf{ to each worker}
        
\State \textbf{end function}
\end{algorithmic}
\end{algorithm}

\begin{algorithm}
\caption{Gradient Aggregation with Gradient Compression and Error Feedback}
\label{alg:gaef}
\begin{algorithmic}[1]
\State \textbf{Input:} compressor $\mathcal{C}$, $g_{t, i}$
\State \textbf{Initialize:} $e_{0, i} = 0$, $\tilde{e}_0 = 0$

\State \textbf{function} compress\_ef\_push\_pull($\mathcal{C}$, $g_i$)

    \State \hspace*{\algorithmicindent} \textbf{On worker i}

    \State \hspace*{\algorithmicindent} \hspace*{\algorithmicindent}
        $q_{t, i} = g_{t, i} + e_{t, i}$
        
    \State \hspace*{\algorithmicindent} \hspace*{\algorithmicindent} \textbf{Push} $\delta_{t, i} = \mathcal{C}(q_{t, i})$ \textbf{ to server}
        
    \State \hspace*{\algorithmicindent} \hspace*{\algorithmicindent} $e_{t + 1, i} = q_{t, i} - \delta_{t, i}$

    \State \hspace*{\algorithmicindent} \textbf{On server}

    \State \hspace*{\algorithmicindent} \hspace*{\algorithmicindent} \textbf{Pull} $\delta_{t, i}$ \textbf{ from each worker}

    \State \hspace*{\algorithmicindent}  \hspace*{\algorithmicindent} $\Delta_t = \frac{1}{n}\sum_{i=1}^n\delta_{t, i} + \tilde{e}_t$
        
    \State \hspace*{\algorithmicindent}  \hspace*{\algorithmicindent} $p_t = \mathcal{C}(\Delta_t)$

    \State \hspace*{\algorithmicindent} \hspace*{\algorithmicindent} \textbf{Push} $p_t$ \textbf{ to each worker}
        
    \State \hspace*{\algorithmicindent} \hspace*{\algorithmicindent}  $\tilde{e}_{t+1} = \Delta_t  -  p_t$   
        
\State \textbf{end function}
\end{algorithmic}
\end{algorithm}

Algorithm~\ref{alg:gagc} shows how we perform two-way gradient compression for the unbiased compressor. First, each worker compresses its local gradient and pushes it to the server. Then, the server aggregates the compressed gradients and compresses the aggregated gradient again before sending it back to all the workers. Examples of unbiased compressors include random-$k$ (random $k$ elements of the gradient are sent) and dithering \cite{alistarh2017qsgd,horvath2019natural} (multi-bit quantization).

Algorithm~\ref{alg:gaef} presents the two-way compression with error feedback. Each worker corrects its local gradient $g_{t, i}$ with the local residual error $e_{t, i}$ from the last iteration before compressing it. Thus, instead of compressing $g_{t, i}$, we compress $q_{t, i} = g_{t, i} + e_{t, i}$. Then, each worker updates its local residual error by computing the compression error: $e_{t+1, i} = q_{t, i} - \mathcal{C}(q_{t, i})$. On the server, we also maintain a residual error $\tilde{e}_{t}$, which stores the compression error of the last push call from the server to workers. We add the residual error $\tilde{e}_{t}$ to the aggregated workers' compressed gradients and then compress it again before pushing it back to all the workers. When $\mathcal{C}$ is an identity mapping, both Algorithms~\ref{alg:gagc} and ~\ref{alg:gaef} recover Algorithm~\ref{alg:ga}. Two commonly adopted biased compressors are scaled sign operator \cite{karimireddy2019error,zheng2019communication} 
and top-$k$ (only $k$ largest gradient elements are sent). The final algorithm, compressed LANS (CLAN), is presented in Algorithm~\ref{alg:clan}. 

\begin{algorithm}
    \caption{Compressed LANS (CLAN)}\label{alg:clan}
\begin{algorithmic}[1]
\State \textbf{Input:} $x_1\in R^d$, learning rate $\{\eta_t\}^T_{t=1}$, parameters $0 < \beta_1, \beta_2 < 1$, scaling function $\phi$, $\epsilon > 0$, compressor $\mathcal{C}$, use$\_$ef $\in$ \{true, false\}

\State \textbf{Initialize:} $m_0 = 0$, $v_0 = 0$

\State \textbf{for} $t=1$ \textbf{to} $T$ \textbf{do} 
    \State \hspace*{\algorithmicindent} {\bf On each worker i}
    \State \hspace*{\algorithmicindent} \hspace*{\algorithmicindent} Compute mini-batch stochastic gradient $g_{t, i}$
    \State \hspace*{\algorithmicindent} \hspace*{\algorithmicindent} \textbf{if not } use$\_$ef \textbf{ then}
            \State \hspace*{\algorithmicindent} \hspace*{\algorithmicindent} \hspace*{\algorithmicindent} $\tg_{t}$ = compress\_push\_pull($\mathcal{C}$, $g_{t, i}$)
    \State \hspace*{\algorithmicindent} \hspace*{\algorithmicindent} \textbf{else} 
            \State \hspace*{\algorithmicindent} \hspace*{\algorithmicindent} \hspace*{\algorithmicindent} $\tg_{t}$ = compress\_ef\_push\_pull($\mathcal{C}$, $g_{t, i}$)  
    \State \hspace*{\algorithmicindent} \hspace*{\algorithmicindent} \textbf{end if}
    \State \hspace*{\algorithmicindent} \hspace*{\algorithmicindent} \textbf{for} $b=1$ \textbf{to} $B$ \textbf{do} 
    
        \State \hspace*{\algorithmicindent} \hspace*{\algorithmicindent} \hspace*{\algorithmicindent} $m_{t, \mG_b} = \beta_1m_{t-1, \mG_b} + (1 - \beta_1)\tg_{t, \mG_b}$
        \State \hspace*{\algorithmicindent} \hspace*{\algorithmicindent} \hspace*{\algorithmicindent} $v_{t, \mG_b} = \beta_2v_{t-1, \mG_b} + (1 - \beta_2)\tg_{t, \mG_b}^2$
        \State \hspace*{\algorithmicindent} \hspace*{\algorithmicindent}  \hspace*{\algorithmicindent} $\tm_{t, \mG_b} = m_{t, \mG_b} / (1 - \beta_1^t)$
        \State \hspace*{\algorithmicindent} \hspace*{\algorithmicindent} \hspace*{\algorithmicindent} $\tv_{t, \mG_b} = v_{t, \mG_b} / (1 - \beta_2^t)$
        \State \hspace*{\algorithmicindent} \hspace*{\algorithmicindent} \hspace*{\algorithmicindent} compute ratios $r_{t, \mG_b} = \frac{\tm_{t, \mG_b}}{\sqrt{\tv_{t, \mG_b}} + \epsilon}$, $c_{t, \mG_b} = \frac{\tg_{t, \mG_b}}{\sqrt{\tv_{t, \mG_b}} + \epsilon}$
        \State \hspace*{\algorithmicindent} \hspace*{\algorithmicindent} \hspace*{\algorithmicindent} 
        $\begin{aligned}[t]
        \tilde{d}_{t, \mG_b} & = \phi(\|x_{t, \mG_b}\|_2)\left[\frac{\beta_1}{\|r_{t, \mG_b} + \lambda x_{t, \mG_b}\|}(r_{t, \mG_b} + \lambda x_{t, \mG_b}) \right.\\
        &\left. + \frac{1 - \beta_1}{\|c_{t, \mG_b} + \lambda x_{t, \mG_b}\|}(c_{t, \mG_b} + \lambda x_{t, \mG_b})\right]
        \end{aligned}$
        \State \hspace*{\algorithmicindent} \hspace*{\algorithmicindent} \hspace*{\algorithmicindent} $x_{t + 1, \mG_b} = x_{t, \mG_b} - \eta_t\tilde{d}_{t, \mG_b}$
        
    \State \hspace*{\algorithmicindent}  \hspace*{\algorithmicindent} \textbf{end for}
\State \textbf{end for}
\end{algorithmic}
\end{algorithm}

\subsection{Convergence Analysis}

Now we establish the theoretical guarantees of the convergence of CLAN. We consider two types of compressors:

\begin{definition} \label{definition:w_compressor}
An operator $\mathcal{C}: \R^d \rightarrow \R^d$ is an $\omega$-compressor for some $\omega \in [0, \infty)$ if $\E[\mathcal{C}(x)] = x$ and 
$\|\mathcal{C}(x) - x\|^2_2 \leq \omega\|x\|_2^2$.
\end{definition}
This is a subclass of strictly $\omega$-compressors proposed in \citep{wen2017terngrad,alistarh2017qsgd,wangni2018gradient}. The examples include random-$k$ and dithering \citep{alistarh2017qsgd,horvath2019natural} operators. The $\omega$-compressors are used when we employ Algorithm~\ref{alg:gagc}.

\begin{definition} \label{definition:ef_compressor}
\citep{karimireddy2019error} 
An operator $\mathcal{C}: \R^d \rightarrow \R^d$ is an $\delta$-approximate compressor for $\delta \in (0, 1]$ if
$\|\mathcal{C}(x) - x\|^2_2 \leq (1 - \delta)\|x\|_2^2$.
\end{definition}
Scaled sign operator 
$\mathcal(v) = \|v\|_1/d\cdot\text{sign}(v)$ \citep{karimireddy2019error,zheng2019communication} 
and top-$k$ operator are two examples. We adopt $\delta$-approximate compressors when we use Algorithm~\ref{alg:gaef}. 
In this section, we establish convergence rates for  Algorithms~\ref{alg:gagc} and ~\ref{alg:clan} with $\omega$-compressors and Algorithms~\ref{alg:gaef} and ~\ref{alg:clan} with $\delta$-approximate compressors, respectively.

We present the convergence analysis for CLAN with either $\omega$-compressors or $\delta$-approximate compressors.

\subsubsection{Assumptions}

We first introduce the following assumptions. Roughly speaking, we use the assumptions similar to LAMB~\citep{you2019large}.

\begin{assumption} (Smoothness)
We assume coordinate-wise smoothness:
$
|(\nabla F(x) - \nabla F(y))_j|
\leq L_j |(x - y)_j|^2, \forall x, y \in \R^d, j \in [d],
$
which implies block-wise smoothness:
\begin{align*}
\|\nabla_{\mG_b} F(x) - \nabla_{\mG_b} F(y)\|
&\leq L_b \|x_{\mG_b} - y_{\mG_b}\|^2, \forall x, y \in \R^d, b \in [B],
\end{align*}
where $L_b = \max_{j \in \mG_b} L_j$.
We define $\|L\|_1 = \sum_{j \in [d]} L_j \geq \sum_{b \in [B]} L_b$.
\label{asm:smoothness}
\end{assumption}

\begin{assumption} (Bounded variance)
For any unbiased stochastic gradient $g_{t}$ such that $\E[ g_t ] = \nabla F(x_{t-1})$, 
we assume coordinate-wise bounded variance:
$
    \E [ (g - \nabla F (x_{t-1}))_j^2 ] \leq \frac{\sigma_j^2}{s}, \forall j \in [d], t \in [T],
$
where $s$ is the batch size.
And we define $\|\sigma\|_1 = \sum_{j \in [d]} \sigma_j$.
\label{asm:variance}
\end{assumption}

\begin{assumption} (Bounded gradients)
We assume coordinate-wisely bounded stochastic gradient $g_{t}$:
$
    (g_{t})_j \leq G, \forall j \in [d], t \in [T].
$
For convenience, we denote $G' = G+\epsilon$.
\label{asm:gradient}
\end{assumption}

\begin{assumption} (Bounded rescaling factor)
We assume the rescaling factor is bounded by:
$
    0 < \alpha_l \leq \phi(\|v\|_2) \leq \alpha_u, \forall v \in \R^d.
$
\end{assumption}

\begin{assumption} (Existence of global minimum)
There exists at least one global minimum $x_*$, where
$
F(x_*) \leq F(x), \forall x \in \R^d.
$
\label{asm:global_min}
\end{assumption}

\subsubsection{Main Results}

Based on the assumptions above, we establish the convergence guarantees of CLAN. For simplicity, throughout the theoretical analysis, we ignore the regularization, i.e., taking $\lambda = 0$, and assuming fixed learning rate $\eta_t = \eta, \forall t \in [T]$. The full proof can be found in the appendix.

Assume that per-worker batch size is $s$ (each $g_{t, i}$ is evaluated on $s$ random samples) and the number of workers is $n$. We consider a general gradient estimator $p_t$ (the block index is ignored) in the $t^{\mbox{th}}$ iteration of CLAN. When using the full precision in pushpull operations, $p_t = \frac{1}{n}\sum_{i \in [n]} g_{t, i}$. When Algorithm~\ref{alg:gagc} is used, we have $p_t = \mathcal{C}\left( \frac{1}{n} \sum_{i \in [n]} \mathcal{C}(g_{t, i}) \right)$, where $\E[p_t] = \frac{1}{n}\sum_{i \in [n]} g_{t, i}$. When Algorithm~\ref{alg:gaef} is used, we have $p_t =  \mathcal{C}\left( \frac{1}{n} \sum_{i \in [n]} \mathcal{C}(g_{t, i} + e_{t, i}) + \tilde{e}_t \right)$.
We first show a general error bound with the undetermined coefficients $V_1, V_2, V_3$ depending on the choice of $p_t$, then provide the specifications of the coefficients for the three options of $p_t$.

\begin{theorem}~\label{thm:clan_general}
After $T$ iterations, for CLAN with the general gradient estimator $p_t$ with $(p_t)_j \leq V_1$, we have the following error bound:
\begin{align}
&\E \| \nabla F(x_o) \|^2
\leq  \frac{\sqrt{d} V'_1 \E[ F(x_{1}) - F(x_{*}) ]}{T \eta \alpha_l (1-\beta_1) \sqrt{1-\beta_2}} 
+ \sqrt{d} G V_3 \nonumber\\
& + \frac{\eta \sqrt{d} V'_1 \alpha_u^2 (1-\beta_1 + 2\beta_1^2) \|L\|_1 }{2 \sqrt{1-\beta_2} (1-\beta_1)^2 \alpha_l} 
+ \frac{\sqrt{d} V'_1 \alpha_u [(1-\beta_1)^2+\beta_1]}{\sqrt{1-\beta_2} (1-\beta_1)^2 \alpha_l} V_2, \label{eq:general_bound}
\end{align}
where $V'_1 = V_1 + \epsilon$, $\sum_{b \in [B]} \sum_{j \in \mG_b} \E |(\nabla_{\mG_b} F(x_t) - p_{t, \mG_b})_j| \leq V_2$, $\|\nabla F(x_t) - \E[p_{t}]\| \leq V_3$, $\forall t \in [T]$, if there is any, and $o$ is randomly selected from $\{1, \dots, T\}$. The exact values of $V_1$, $V_2$, and $V_3$ depend on the choice of the gradient estimator $p_t$.
\end{theorem}

The theorem above provides the general error bound of CLAN. Based on it, we present the following three corollaries for CLAN with full precision~(no compression), Algorithm~\ref{alg:gagc} with $\omega$-compressors, and Algorithm~\ref{alg:gaef} with $\delta$-approximate compressors, respectively.

\begin{corollary}~\label{cor:clan_full_precision}
For CLAN with full precision~(i.e., LANS), we have $p_t = \frac{1}{n} \sum_{i=1}^n g_{t, i}$, $V'_1 = G' = G + \epsilon$, $V_2 = \frac{\|\sigma\|_1}{\sqrt{n s}}$, $V_3 = 0$.
Taking  $\eta = \frac{1}{\sqrt{T}}$, and $s = \frac{T}{n}$, we have the following convergence rate: 
\begin{align*}
&\E \| \nabla F(x_o) \|^2 \\
&\leq \mathcal{O}\left( \frac{\sqrt{d} G'}{(1-\beta_1) \sqrt{1-\beta_2}} \times \frac{ [F(x_{1}) - F(x_{*})] + \|L\|_1 + \|\sigma\|_1 }{\sqrt{T}}  \right).
\end{align*}
\end{corollary}
The corollary shows that full-precision CLAN converges to a stationary point at a rate of $\mathcal{O}(1/\sqrt{T})$ for large batch size $s = n / T$. This is consistent with the existing works \citep{reddi2019convergence,zaheer2018adaptive,you2019large} that the optimization methods with Adam-like update for the second-order momentum $v_t$ do not converge in general with small batch size. A much weaker increase in batch size is sufficient for many machine learning applications of interest in practice. As long as we have a small enough gradient variance, a moderately large batch size can work well. The following corollary shows that we also achieve an $\mathcal{O}(1/\sqrt{T})$ rate for CLAN with $\omega$-compressors.

\begin{corollary}~\label{cor:clan_unbiased}
For CLAN with Algorithm~\ref{alg:gagc} and $\omega$-compressors, we have $p_t = \mathcal{C}\left( \frac{1}{n} \sum_{i \in [n]} \mathcal{C}(g_{t, i}) \right)$, where $\E[p_t] = \frac{1}{n}\sum_{i \in [n]} g_{t, i}$, 
$V'_1 = (1 + d \sqrt{ 4 \omega^2 + 6\omega }) G + \epsilon$, $V_2 = \frac{\|\sigma\|_1}{\sqrt{n s}} + dG \sqrt{ 4 \omega^2 + 6\omega }$, $V_3 = 0$.
Taking $\eta = \frac{1}{\sqrt{T}}$, $s = \frac{T}{n}$, and $\omega \leq \frac{1}{T}$, we have the following convergence rate: 
\begin{align*}
&\E \| \nabla F(x_o) \|^2 \\
&\leq \mathcal{O}\left( \frac{\sqrt{d} V'_1}{(1-\beta_1) \sqrt{1-\beta_2}} \times \frac{ [F(x_{1}) - F(x_{*})] + \|L\|_1 + \|\sigma\|_1 + d G }{\sqrt{T}}  \right).
\end{align*}
\end{corollary}
It can be seen that the term ($V_2$) in (\ref{eq:general_bound}) decreases with more workers $n$ and larger batch size $s$. As Adam-like updates do not converge in general, we need additional assumptions for batch size and compressor parameter $\omega$ to achieve a convergence rate of $\mathcal{O}(1/\sqrt{T})$. 

\begin{corollary}~\label{cor:clan_biased}
For CLAN with Algorithm~\ref{alg:gaef} and $\delta$-approximate compressors, we have $p_t =  \mathcal{C}\left( \frac{1}{n} \sum_{i \in [n]} \mathcal{C}(g_{t, i} + e_{t, i}) + \tilde{e}_t \right)$,
$
V_1 = G + \sqrt{d} V_3
$, 
$
V_2 =
\frac{\|\sigma\|_1}{\sqrt{n s}} + \sqrt{d} V_3
$,
$V_3 =
\frac{2 \sqrt{1-\delta}}{1-\sqrt{1-\delta}} \left[ \sqrt{d} + 2 \left(1+\frac{\sqrt{d(1-\delta)}}{1 - \sqrt{1-\delta}}\right) \right] G
$.
Taking $\eta = \frac{1}{\sqrt{T}}$, $s = \frac{T}{n}$, and $\delta \geq 1 - \frac{1}{(\sqrt{T}-1)^2}$, we have the following error bound: 
\begin{align*}
&\E \| \nabla F(x_o) \|^2 
\leq \mathcal{O}\left( \frac{\sqrt{d} (G + \sqrt{d} V_3 + \epsilon)}{(1-\beta_1) \sqrt{1-\beta_2}} \times \frac{ [F(x_{1}) - F(x_{*})]}{\sqrt{T}}  \right)\\
&\quad + \mathcal{O}\left( \frac{\sqrt{d} (G + \sqrt{d} V_3 + \epsilon)}{(1-\beta_1) \sqrt{1-\beta_2}} \times \frac{ \|L\|_1 + \|\sigma\|_1 + d G }{\sqrt{T}}  \right)
+ \mathcal{O}\left( \frac{d G^2}{\sqrt{T}} \right).
\end{align*}
\end{corollary}
Similarly, we get smaller $V_2$ in (\ref{eq:general_bound}) by using more workers and larger batch size. Moreover, we need extra assumptions for batch size and compressor parameter $\delta$ to ensure global convergence.

In practice, we can have good empirical results by using the compressors with very high compression rates, though theoretically we need to limit the compression rate to achieve small enough approximation factors $\omega$ and $\delta$ so that we have $\mathcal{O}(1/\sqrt{T})$ convergence rate for both $\omega$ and $\delta$-approximate compressors. Note that the error terms caused by compression are also controlled by $G$, which upper bounds the coordinates of the stochastic gradients. Thus, as long as the gradients are small enough, the compressors with high compression rates can achieve good convergence. Furthermore, the initial function gap $[F(x_{1}) - F(x_{*})]$ dominates the training at the beginning, which makes the compression error less critical. Especially for optimizing deep neural networks, the convergence of ``transient phase'' \cite{darken1991towards} matters a lot more than the asymptotic convergence \cite{sutskever2013importance}.
On the other hand, when approaching the critical points at the end of the training, the gradients are supposed to get smaller, resulting in more minor compression errors. Thus, in either case, we can still use relatively high compression rates and achieve good convergence.

\begin{figure}[tp]
    \vskip -.1in
    \centering
    \includegraphics[scale=0.34]{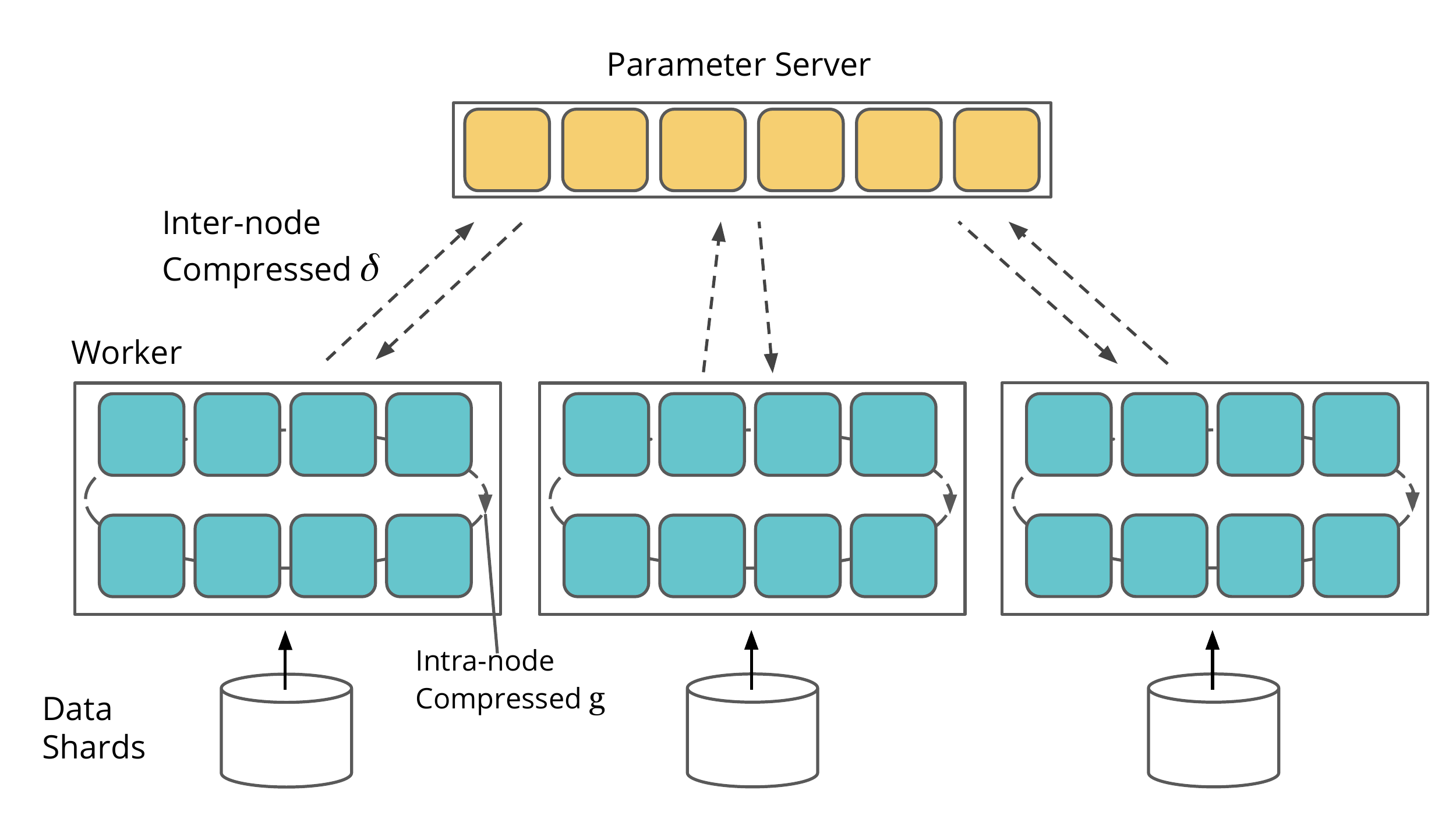}
    \caption{The system architecture of BytePS-Compress. The ring through GPUs represents an All-Reduce operation on intra-node compressed gradients. The arrows between workers and parameter servers represent Push and Pull operations on inter-node compressed gradients.}
    \label{fig:1}
    \vskip -.15in
\end{figure}

\begin{table}[h]
    \centering
    \caption{Communication volume with regard to number of workers ($n$).}
    \vskip -.1in
    \begin{tabular}{|l|c|}
        \hline Primitive & Communication Volume \\
        \hline All-Gather $/$ Broadcast & $\mathcal{O}(\mathrm{n})$ \\
        \hline All-Reduce & $\mathcal{O}(1)$ \\
        \hline Push $/$ Pull & $\mathcal{O}(1)$\\
        \hline
    \end{tabular}
    \label{tab:comm-primitives}
    \vskip -.15in
\end{table}

\section{BytePS-Compress: A Scalable System for Gradient Compression}

Though gradient compression has a solid theoretical guarantee, its overhead is non-trivial and can slow down the training with poor system design. Recently, GRACE \cite{xu2020compressed} was introduced as a gradient compression framework based on Horovod \cite{sergeev2018horovod}. However, their system is not scalable and not well optimized. Most of the compressors are implemented with All-Gather and Broadcast primitives for communication, whose communication costs increase linearly with the number of workers (see Table~\ref{tab:comm-primitives}), resulting in poor scalability. As a result, their experiments show that gradient compression slows down the training. In this paper, BytePS-Compress, whose communication cost remains constant with varying numbers of nodes, is described in more detail below.


\subsection{System Design}
BytePS-Compress is built on BytePS \cite{peng2019generic} and employs a two-stage communication: 1) intra-node communication that performs All-Reduce gradients across GPUs with NVIDIA’s Collective Communication Library (NCCL); 2) inter-node compression that exchanges gradients with parameter servers through TCP or Remote Direct Memory Access (RDMA) network. In this work, we use different compression schemes for intra-node and inter-node communications, respectively, which are aligned with the capabilities/characteristics of the underlying architectures.

\subsubsection{Intra-Node Compression} 

High-performance training servers such as Amazon's P3.16xlarge and NVIDIA's DGX-1 are equipped with high-speed interconnects such as PCI-E or NVlink/NVSwitch. Therefore, intra-node communication is relatively fast compared to inter-node communication over Ethernet. In addition, applying quantization and sparsification methods to All-Reduce will lead to significantly higher latency. This is because the compressed data cannot be summed directly, requiring recursive decompression, aggregation, and compression. In light of these, we use simple data type conversion such as FP32 to FP16 for intra-compression, which has a low computation overhead.


\subsubsection{Inter-Node Compression}
Different from intra-node compression, inter-node communication suffers much more from the limited network bandwidth. Thus, we consider using compression algorithms such as 1-bit \cite{bernstein2018signsgd} and top-$k$ \cite{stich2018sparsified} that achieve a high compression rate to reduce the number of bits for the inter-node data transfers. In this stage, we use the parameter-server architecture and two-way compression, where the gradients are compressed before being sent to servers. Then, the aggregated gradients are compressed again before getting pulled from the workers. 

The architecture of the proposed BytePS-Compress is shown in Figure \ref{fig:1}. Gradients are first compressed by intra-node compression on each worker and then reduced across GPUs via All-Reduce within a host. Then the local average gradients are further compressed by inter-node compression before transmitting them over the network. Next, servers merge the gradients and compress the aggregated gradient again before responding to the pull requests from workers. 
Finally, workers decompress the compressed aggregated gradients and broadcast them to all GPUs. This is the life-cycle of a compressed gradient push-pull task for a mini-batch. Thanks to the usage of {\bf push\_pull} operations across nodes, the inter-node communication cost remains constant as the number of nodes grows, making the system easier to scale. 

The inter-node compressors run on the CPU. We opt for CPU compressors instead of GPU compressors for the following reasons. Firstly, with the parameter-server architecture, one could add additional CPU nodes as servers to utilize more networking bandwidth, where GPUs may not be available. Secondly, GPU kernels are less flexible than CPUs. Some operations such as top-$k$ cannot be easily parallelized on the GPU, resulting in significant compression overhead and slow down the training time \cite{gupta2020fast}. Last but not least, CPU resources are plentiful. It has been observed that the average CPU utilization during the training is only around 20\% - 35\% \cite{jiang2020unified}. Since CUDA kernel calls are asynchronous, the compression on CPUs does not interfere with the GPU execution. Thus, we utilize CPU resources for inter-node compression and decompression.

\subsection{Reduce Compression Overhead} 
\label{sec:rco}
\subsubsection{Parallel CPU Compressors}
We design two types of parallelism. Firstly, we asynchronously launch dozens of compression and decompression jobs to achieve inter-task parallelism. These jobs run simultaneously on different CPU threads. Thus, compared to only executing one kernel at a time on a GPU, the degree of parallelism on CPUs is much higher. Secondly, for each job, SIMD instructions and OpenMP are used to achieve intra-task parallelism. 

\subsubsection{Operator Fusion} 
As shown in Algorithm~\ref{alg:gaef}, to update residual error $e_{t, i}$ (resp. $\tilde{e}_{t, i}$), we need to decompress the data which has just been compressed and then subtract it from the corrected gradient $q_{t, i}$ (resp. $p_t$), which incurs an $\mathcal{O}(d)$ computational complexity. To alleviate the cost, we introduce an efficient way to bypass such decompression for certain compressors. Take top-$k$ and random-$k$ for example. Both algorithms select $k$ elements and corresponding indices, leaving other elements as zeros. Thus, we copy the corrected gradient to the error buffer and then zero-fill the selected $k$ elements. The computational complexity is thus reduced to $\mathcal{O}(k)$, where $k \ll d$.

\subsubsection{Size Threshold} 
For the small tensors, the gradient compression could harm performance due to a constant overhead. Therefore, we set a threshold to filter out small tensors. Users can control this threshold. We leave the automatic threshold search as future work. Currently, we set the threshold to 1MB by default.

\subsubsection{Workload Balance}
For the full-precision training, the workload is evenly distributed across all computing resources on servers. However, with the size threshold, the workload is unbalanced as some gradients demand extra computation because of the compression and decompression. Thus, we allocate more computing power to the gradients that will get compressed to balance the workload.

\subsubsection{More Servers} 
As servers need to spend some resources on compression, we can use more servers to reduce the per-server workload. Note that increasing the number of servers will also benefit the full-precision method, but the benefits are more significant for gradient compression. In practice, we place two servers on each worker node without spinning off any additional CPU machines. 

\subsubsection{NUMA Tuning} 
Deep learning tasks launched by BytePS will spawn one subprocess per GPU device, and each subprocess will have multiple threads to maintain its pipeline. These threads may not be scheduled on the same NUMA node simultaneously, incurring a cross-node memory access overhead. We resolve the issue by allocating a fixed set of CPUs to each subprocess. The allocation algorithm is static, and more CPUs are assigned to the root subprocess because it is responsible for both the compression and network transmission.

\section{Experiments}

This section presents empirical results on two tasks: image classification on ImageNet \cite{russakovsky2015imagenet} and BERT pretraining \cite{devlin2018bert} on Wikipedia and BookCorpus. We first demonstrate the efficiency of BytePS-Compress in training two different convolutional neural networks with state-of-the-art gradient compression methods. Secondly, we evaluate the performance of the proposed CLAN in BERT pretraining.

Experiments are conducted on Amazon EC2 P3.16xlarge instances, each equipped with 8 NVIDIA V100 GPUs and 25Gbps Ethernet. Inside each host, the GPUs are interconnected with high-speed NVLink. We also add more servers for a fair comparison and use NUMA tuning for the baseline in the following experiments. 

\subsection{ImageNet}

We evaluate the performance of BytePS-Compress on the training on ImageNet \cite{russakovsky2015imagenet} with two representative CNN models - ResNet50 \cite{he2016identity} and VGG16 \cite{simonyan2014very}. We use the training scripts in Gluon-CV toolkit \cite{guo2020gluoncv}, with a small modification to accommodate gradient compression.

We compare several state-of-the-art methods in our experiments: 1) full-precision Nesterov accelerated gradient (NAG) \cite{sutskever2013importance}; 2) half-precision NAG (FP32 gradients are converted to FP16 during the communication); 3) scaled 1-bit with error feedback (EF) \cite{zheng2019communication}; 4) random-$k$ with EF ($k=1/32$), which drops 96.875\% gradients \cite{stich2018sparsified}; 5) top-$k$ with EF ($k=0.1\%$) \cite{stich2018sparsified}, which drops $99.9\%$ gradients; 6) linear dithering (5 bits) \cite{alistarh2017qsgd}; 7) natural dithering (3 bits) \cite{horvath2019natural}. All the compression methods are applied to NAG. 

\subsubsection{Workload Breakdown}  In this experiment, we measure the workload breakdown into computation and communication. We first measure the computation time by collecting the elapsed time of running 50 iterations ($t_0$) on one node. Then we measure the total training time for running 50 iterations ($t_1$) on 8 nodes. Then, we get an estimate of communication time using $t_1 - t_0$. Note that we also consider compression time as part of communication time.

As Figure \ref{fig:2-1} shows, gradient compression  reduces communication to varying degrees. For ResNet50, the drop is trivial, up to 11.1\%, mainly due to the smaller model size. In contrast, a remarkable decline (79.0\% decrease) occurs using random-$k$ for VGG16 since VGG16 has a larger model size (528MB). We notice that the communication time of ResNet50 increases with FP16 compression. Due to lack of hardware support, it is relatively slow to do FP16 summation on CPU, which results in a longer training time. 

\begin{figure}[t]
    \centering
    \includegraphics[scale=0.20]{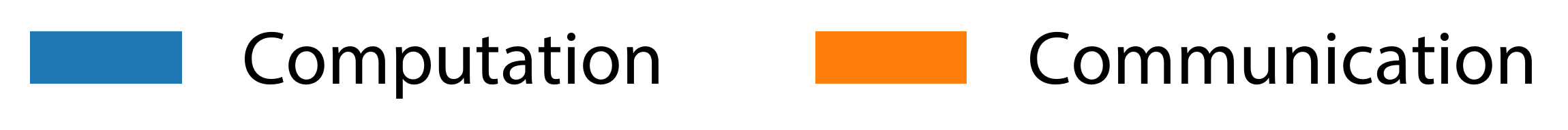}

    \includegraphics[scale=0.112]{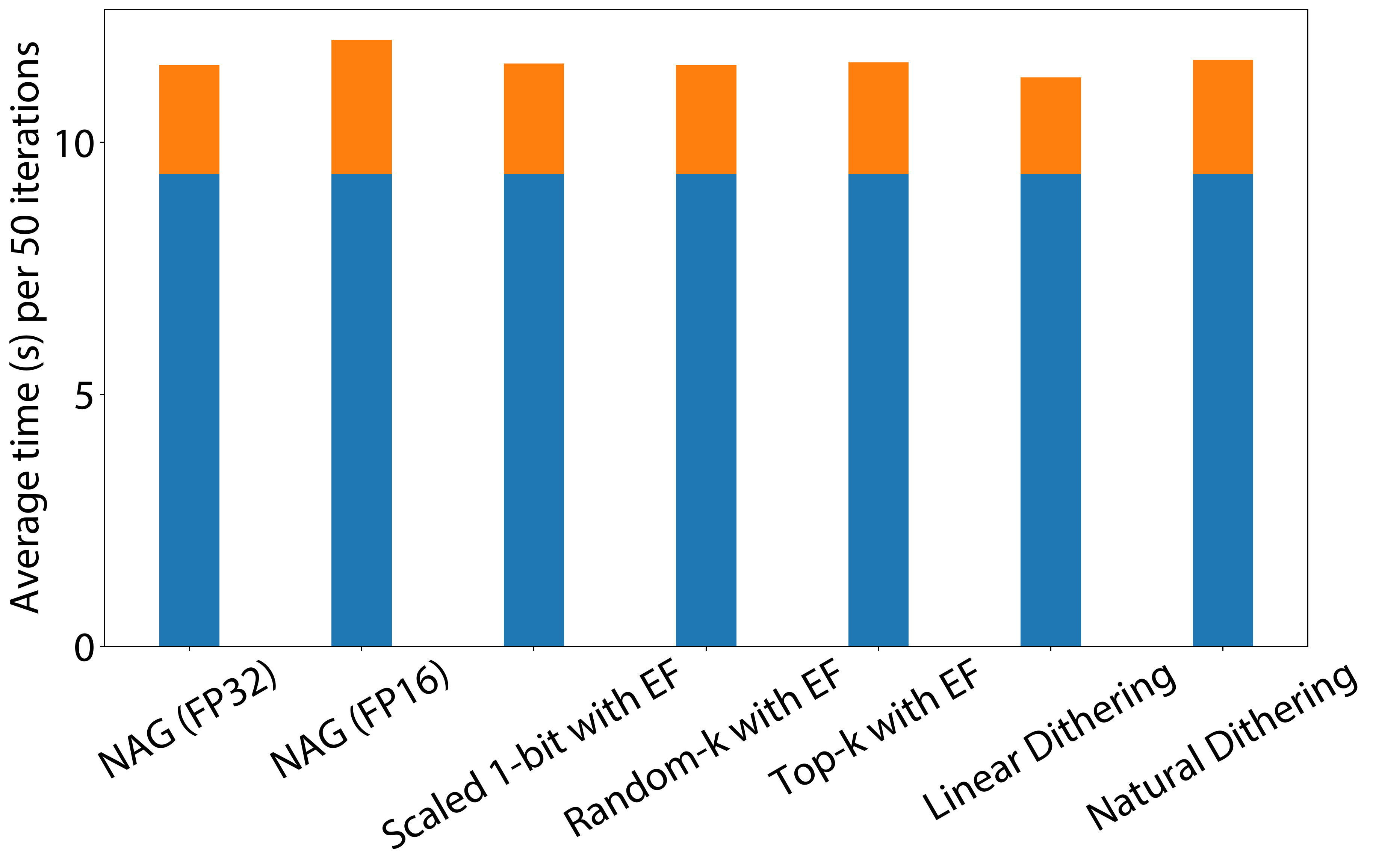}
    \includegraphics[scale=0.112]{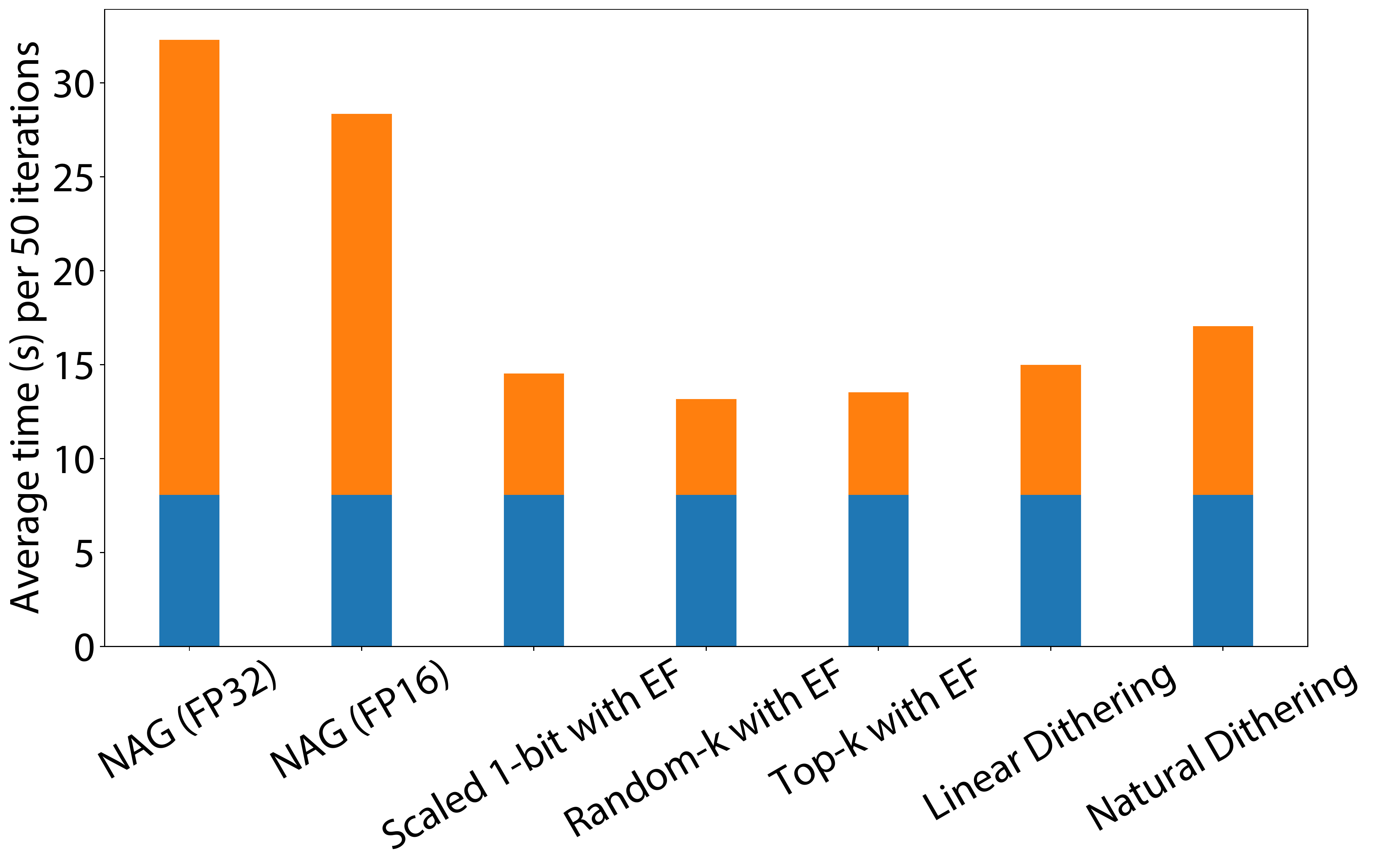}
    \caption{Workload breakdown into computation and communication. Left: ResNet50; Right: VGG16}
   \label{fig:2-1}
   \vskip -0.15in
\end{figure}

\begin{figure}[t]
    \centering
    \includegraphics[scale=0.18]{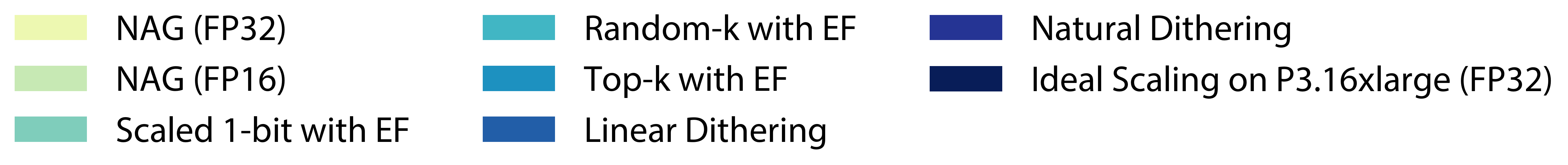}
    \includegraphics[scale=0.125]{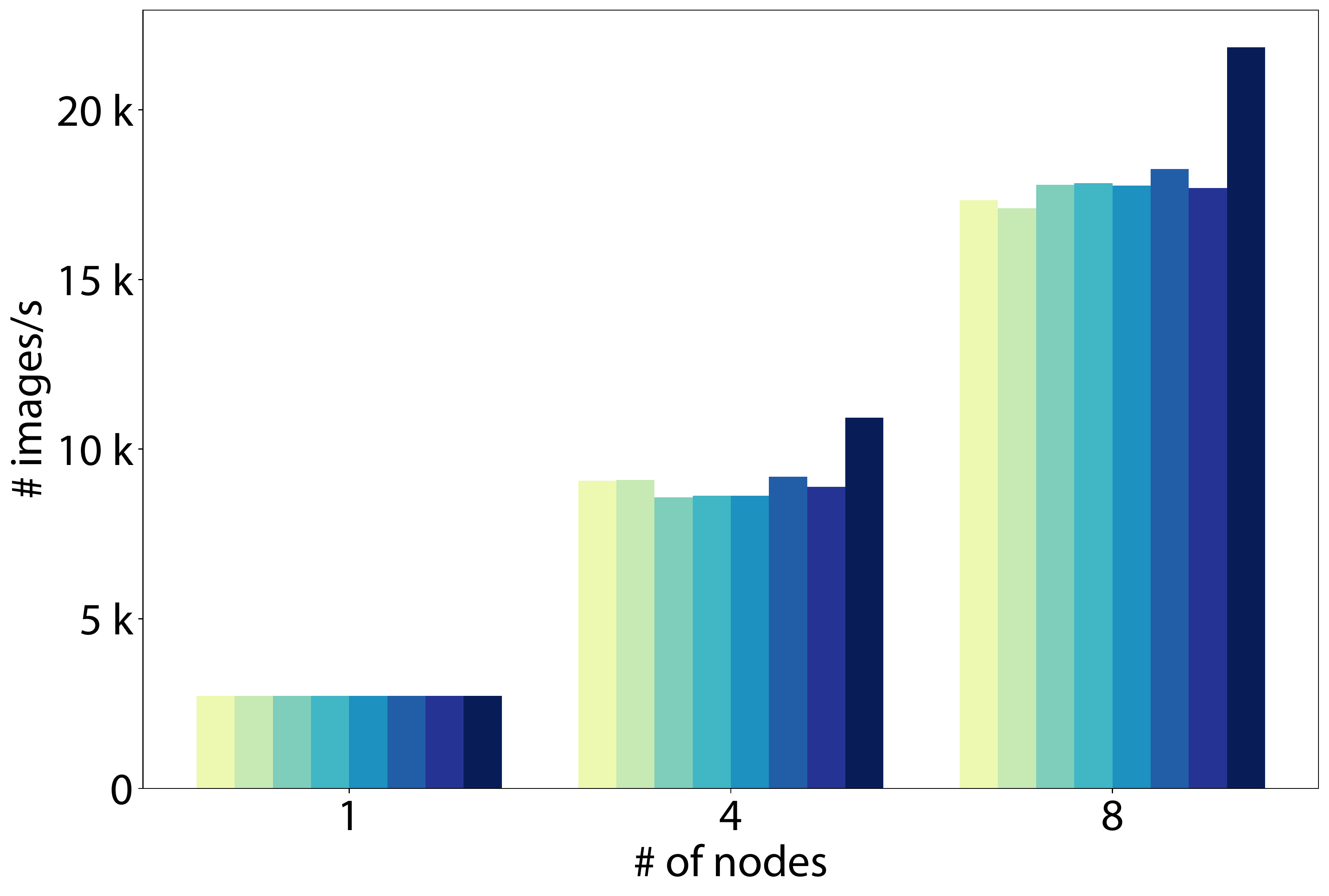}
    \includegraphics[scale=0.125]{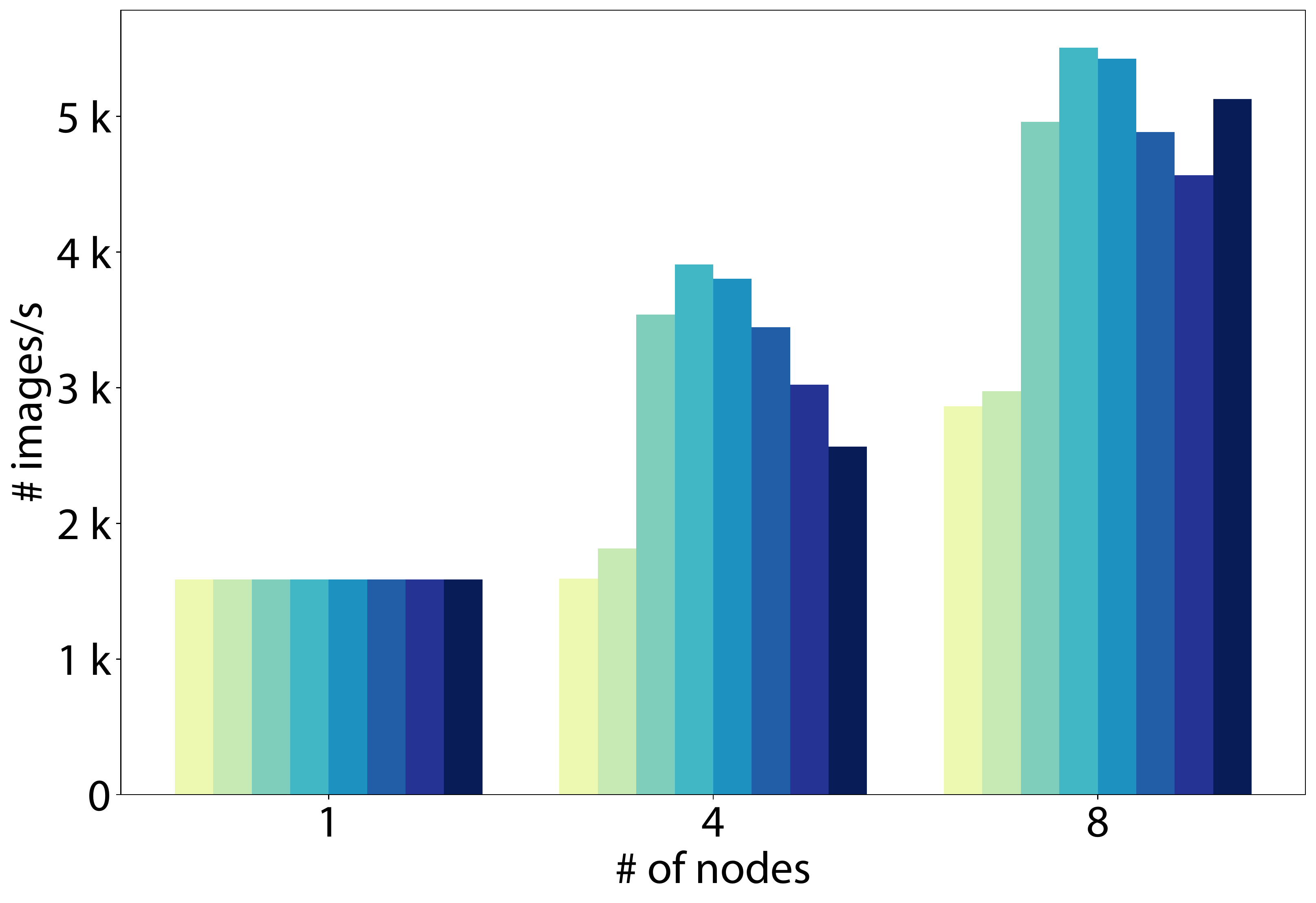}
    \caption{Throughput scaling by increasing the number of nodes from 1 to 8. Left: ResNet50; Right: VGG16}
   \label{fig:2-2}
   \vskip -0.15in
\end{figure}

\subsubsection{Scaling Efficiency} We also measure scaling efficiency when the nodes varies from 1 to 8, with PS workers and servers co-located on each node.
The ideal scaling efficiency is calculated according to the following formula: 

\begin{equation*}\label{scaling}
    scale_{ideal} = \frac{T_{FP} + T_{BP}}{T_{FP} + \max(T_{BP}, T_{COMM})},
\end{equation*}
where $T_{FP}$ and $T_{BP}$ are the time taken by forward pass and backward propagation, respectively, and $T_{COMM} = 2d / bandwidth$ is the theoretical communication cost of parameter-server on Amazon EC2 P3.16xlarge instances. It is assumed that communication operations and gradient computation are perfectly overlapped. Thus, in our setting, ResNet50 can reach 100\% ideal scaling efficiency, while VGG16 has only an ideal scaling efficiency of 40.4\%.

Figure \ref{fig:2-2} shows that gradient compression improves the scaling efficiency over the full-precision method. The efficiency gain in gradient compression is much higher for VGG16 than ResNet50, since ResNet50 has a more negligible communication overhead. Note that scaling efficiency after the gradient compression can exceed the ideal one. The message size is reduced, and the actual communication cost can be smaller than the theoretical one.

\subsubsection{End-to-end Training} Finally, we train ResNet50 and VGG16 end-to-end to measure the total reduction in training time. We follow the linear scaling schedule proposed in \citeauthor{goyal2017accurate}. We set the number of warm-up epochs to 5 and train for 120 epochs (100 epochs for VGG16). We train ResNet50 and VGG16 on 8 and 4 nodes, respectively. We observe that VGG16 suffers from the generalization gap \cite{goyal2017accurate} on 8 nodes even with full-precision NAG, so for VGG16 we use 4 nodes instead.

\begin{figure}[ht]
    \vskip -0.05in
    \centering
    
     \begin{subfigure}[b]{0.47\textwidth}
         \centering
         \includegraphics[scale=0.14]{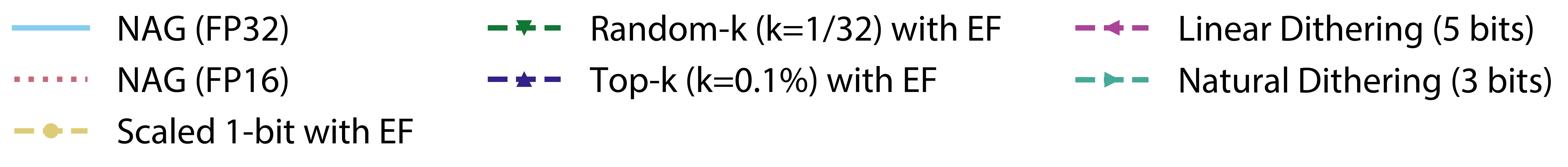}
     \end{subfigure}

     \begin{subfigure}[b]{0.235\textwidth}
         \centering
         \includegraphics[scale=0.123]{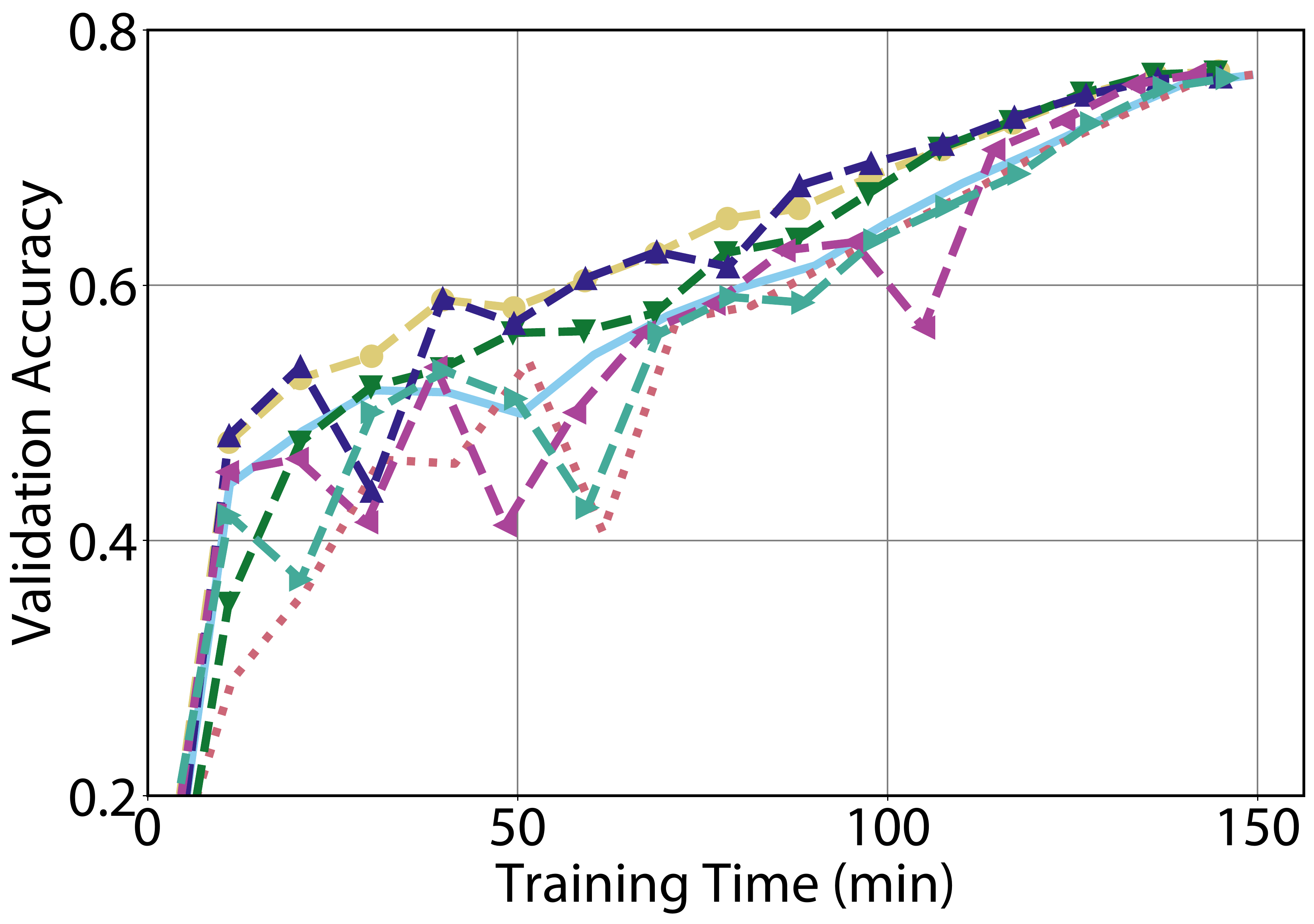}
         \label{fig:resnet50_time_accu}
     \end{subfigure}
     \hfill
     \begin{subfigure}[b]{0.235\textwidth}
         \centering
         \includegraphics[scale=0.123]{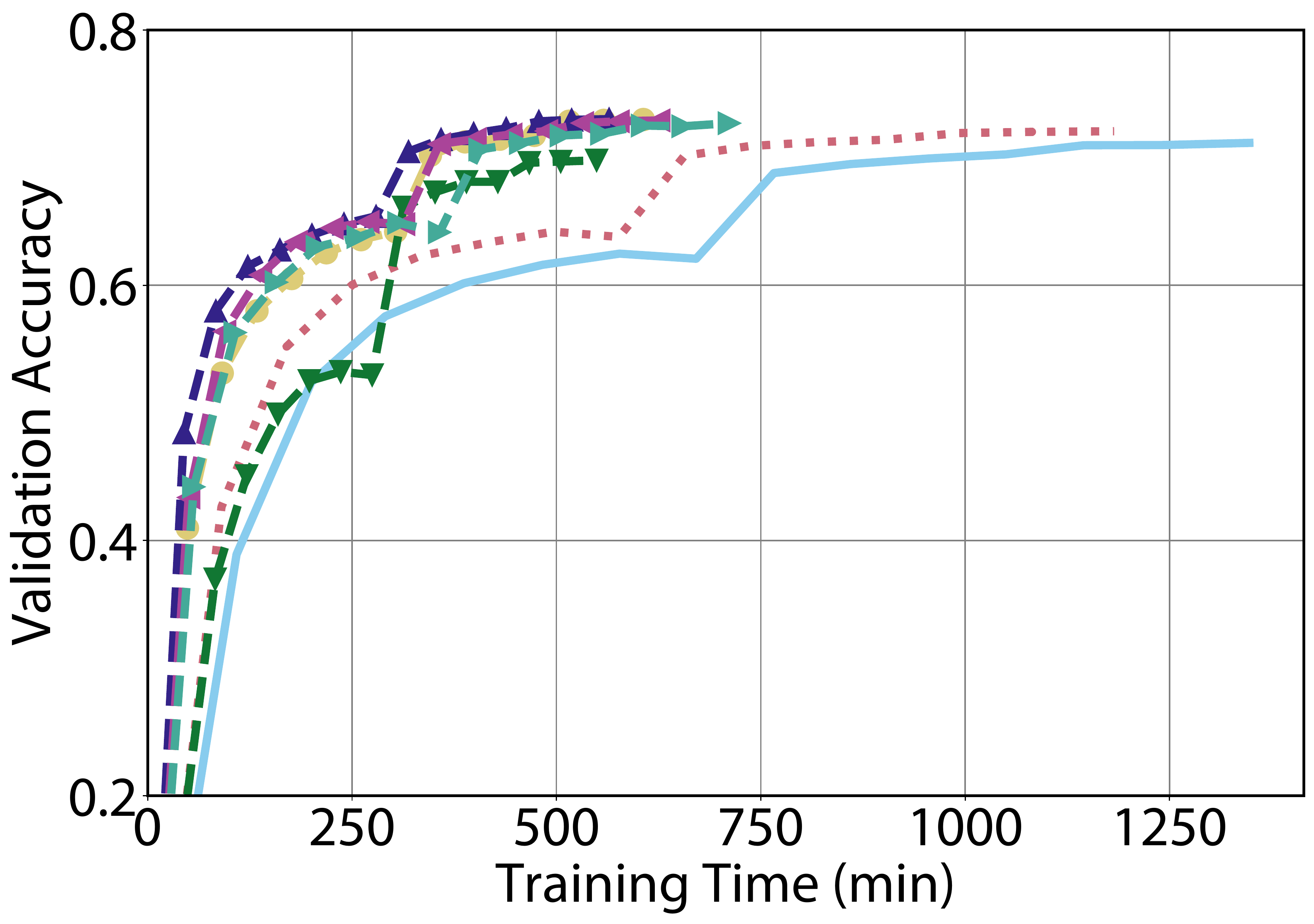}
         \label{fig:vgg16_time_accu}
     \end{subfigure}
  \vskip -0.2in
  \caption{Distributed training results on ImageNet. Left: ResNet50; Right: VGG16.}
  \label{fig:2-3-1}
  \vskip -0.15in
\end{figure}

As shown in Figure \ref{fig:2-3-1}, gradient compression achieves trivial performance gain (5.0\%) on ResNet50 without loss of accuracy. At the same time, all compressors reduce the training time by a large margin compared to the full-precision NAG on VGG16. The complete results in terms of training time and testing accuracy are shown in Table~\ref{tbl:resnet-vgg16}. Except for random-$k$ on VGG16 training, the gradient compression methods achieve the same or even higher accuracy than the full-precision baseline in all cases. In particular, top-$k$ reduces the end-to-end training time of VGG16 by 58.1\% with $1.91\%$ higher accuracy than the full-precision counterpart. On the other hand, although random-$k$ is the fastest algorithm, it fails to achieve the same accuracy as the full-precision NAG for VGG16.



\begin{table}[ht]
\vskip -0.05in
\caption{ End-to-end distributed training for ResNet50 and VGG16 on ImageNet on 8 and 4 Amazon EC2 P3.16xlarge, respectively.}
\vskip -0.05in
\label{tbl:resnet-vgg16}
\centering 
 \begin{tabular}{|l|c|c|c|c|} 
 \hline
 Algorithm & \multicolumn{2}{c|}{ResNet50} & \multicolumn{2}{c|}{VGG16} \\
 \hline
 & Acc & Time & Acc & Time \\
 \hline
 NAG & 76.47 & 148.88~\textrm{m} & 71.16 & 1347.59~\textrm{m} \\
 \hline
 NAG (FP16) & 76.67 & 150.68~\textrm{m} & 72.21 & 1182.21~\textrm{m} \\
 \hline
 Scaled 1-bit with EF & \textsb{76.86} & 144.72~\textrm{m} & 73.04 & 606.35~\textrm{m} \\ 
 \hline 
 Random-$k$ with EF & 76.72 & 144.34~\textrm{m} & 69.84 & \textsb{549.15~\textrm{m}} \\ 
 \hline
 Top-$k$ with EF & 76.46 & 145.00~\textrm{m} & \textsb{73.07} & 564.42~\textrm{m}  \\
 \hline
 Linear Dithering & 76.54 & \textsb{141.51~\textrm{m}} & 72.92 & 625.19~\textrm{m} \\
 \hline
 Natural Dithering & 76.30 & 145.92~\textrm{m} & 72.70 & 711.61~\textrm{m} \\
 \hline
\end{tabular}
\vskip -0.15in
\end{table}

\subsection{BERT Pretraining}

In this experiment, we present the empirical results for speeding up BERT pretraining. We train a BERT-base model on Wikipedia and BooksCorpus datasets on 4 Amazon EC2 P3.16xlarge instances, which amount to 32 GPUs in total. Following the setup in \cite{devlin2018bert}, the pretraining is divided into two stages: a sequence length of $128$ is used in the first $90\%$ training steps, and a sequence length of $512$ is used in the last $10\%$ training steps. The pretraining objective of BERT consists of masked language model (MLM) and next sentence prediction (NSP). We train the model for $250,000$ iterations in total. We use a batch size of $2048$.

For BERT-base, we compare four algorithms: 1) mixed-precision LANS (model is trained using mixed-precision); 2) CLAN with scaled 1-bit and EF; 3) CLAN with top-$k$ and EF ($k=0.1\%$), which drops $99.9\%$ gradients. 4) CLAN with linear dithering (7 bits). Note that for top-$k$, we also need to send the indices of the gradients, which are represented by the int32. Thus, top-$k$ achieves a compression rate of 333x. 

\begin{figure}[ht]
\vskip -0.05in
    \centering
    \includegraphics[scale=0.15]{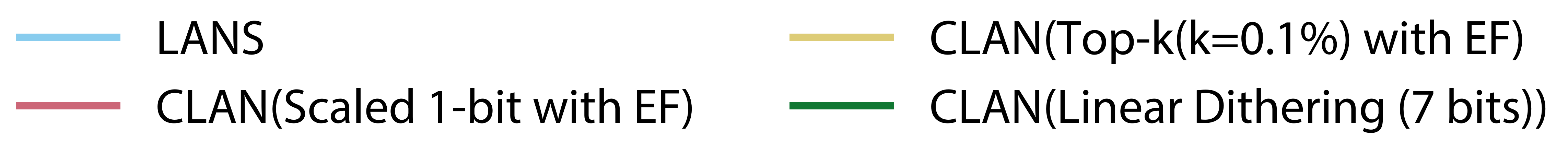}
    
    \includegraphics[scale=0.125]{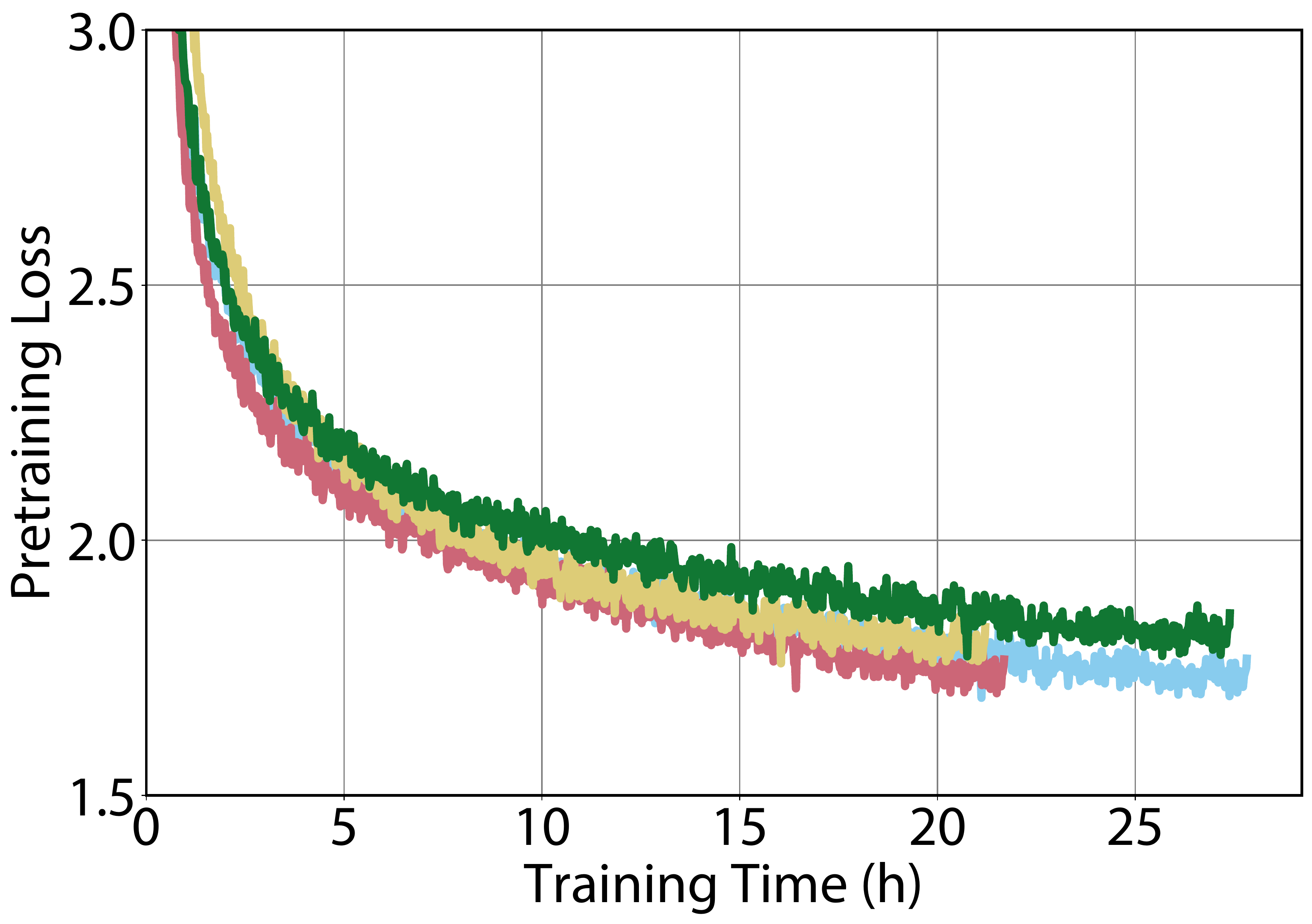}
    \includegraphics[scale=0.125]{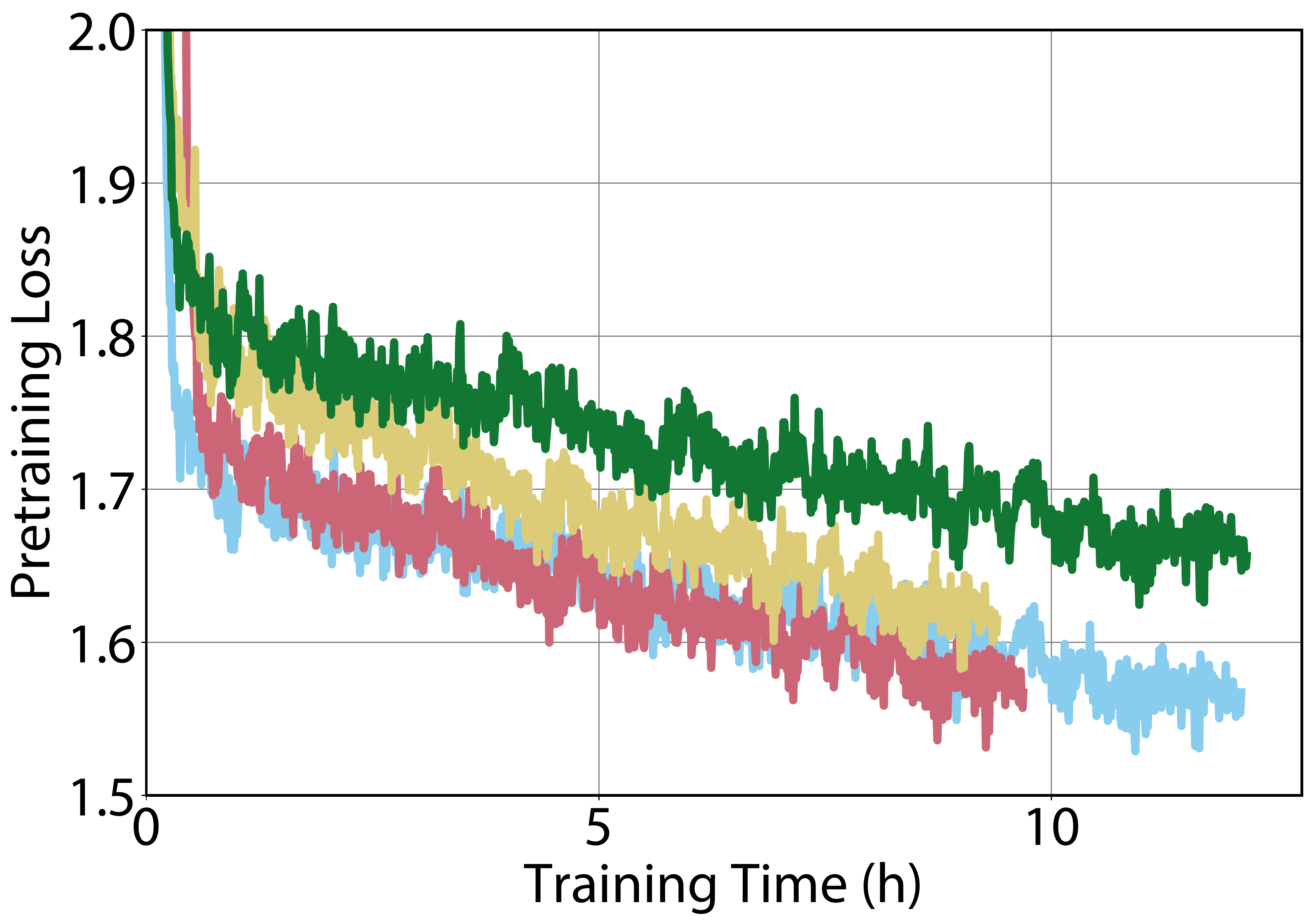}
    \vskip -0.05in
    \caption{End-to-end training performance with BERT-base.  Left: pretraining loss against time in phase 1; Right: pretraining loss against time in phase 2}
   \label{fig:pretraining}
   \vskip -0.15in
\end{figure}

\begin{table}[hbt]
    \centering
    \caption{Experiment results on BERT pretraining. The F1 score on SQuAD-v1.1 development set is used as the evaluation metric. We run the experiment five times and take the mean and standard deviation.}
    \vskip -0.05in
    \begin{tabular}{|l|c|c|}
        \hline Algorithm & F1 score & Pretraining Time \\
        \hline \cite{devlin2018bert} & 88.5 & \strike{|c|}{} \\
        \hline LANS & \textsb{89.3} & 39.9~\textrm{h} \\
        \hline CLAN (Top-$k$ with EF) & 89.2 & \textsb{30.6~\textrm{h}} \\
        \hline CLAN (Scaled 1-bit with EF) & 89.2 & 31.4~\textrm{h}\\
        \hline CLAN (Linear Dithering) & 88.8 & 39.6~\textrm{h}\\
        \hline
    \end{tabular}
    \label{tbl:bert-base}
    \vskip -0.1in
\end{table}

The pretraining curves are shown in Figure~\ref{fig:pretraining}. As it can be seen that CLAN with both scaled 1-bit and top-$k$ show the same convergence performance as LANS, while linear dithering is slightly worse. The evaluation results on the SQuAD-v1.1 development set are shown in Table~\ref{tbl:bert-base}. We also quote the result reported in \cite{devlin2018bert} for comparison. The F1 score is averaged over $5$ runs of the finetuning to reduce the statistical variance. As can be seen that CLAN with error feedback and either top-$k$ or scaled 1-bit matches the performance of LANS, and both methods reduce the pretraining time by more than $20\%$ and outperform the result reported in \cite{devlin2018bert}. However, CLAN with linear dithering has a small gap with its mixed-precision counterpart.

\begin{table}[h]
    \vskip -0.05in
    \centering
    \caption{Experiment results on 4 GLUE tasks. The accuracy on the development set is used as the evaluation metric.}
    \vskip -0.05in
    \begin{tabular}{|l|c|c|c|c|}
    \hline Algorithm & MNLI-m & QNLI & SST-2 & MRPC\\
    \hline \cite{devlin2018bert} & \textsb{84.4} & 88.4 & \textsb{92.7} & \textsb{86.7}  \\
    \hline LANS  & 83.6 & \textsb{90.7} & 91.9 & 86.6  \\ 
    \hline CLAN (Top-$k$ with EF) & 83.3 & 90.5 & 91.9 & 85.2  \\
    \hline CLAN (Scaled 1-bit with EF)  & 83.5 & 90.3 & 92.4 & \textsb{86.7}  \\ 
    \hline CLAN (Linear Dithering) & 83.2 & 90.5 & 91.3 & 84.4 \\ 
    \hline
    \end{tabular}
    \label{tbl:bert-glue}
    \vskip -0.1in
\end{table}

We also evaluate performance on four GLUE tasks as shown in Table~\ref{tbl:bert-glue}. Following \cite{devlin2018bert}, we use a batch size of 32 and finetune for 3 epochs (with an exception for MRPC, which is tiny and where we use 5 epochs). And we search for the best finetuning learning rate among $\{2\times10^{-5}, 3\times10^{-5}, 4\times10^{-5}, 5\times10^{-5}\}$. We report the average accuracy on the development set from 5 random starts of fine-tuning. We also quote the results reported in \cite{devlin2018bert} for comparison. As can be seen, CLAN with scaled 1-bit matches the performance of LANS, while top-$k$ loses a bit of accuracy on MRPC.

\subsubsection{System Scalability} 
We then evaluate the scalability of BytePS-Compress with different scales of the BERT model. Since top-$k$ has the lowest training time, as shown in Table~\ref{tbl:bert-base}, we use it to evaluate the scalability of the system. The results are presented in Table~\ref{tbl:bert-throughput}. In addition to BERT-base and BERT-large, we include the throughput of training a 437M parameter BERT model, which is obtained by increasing the number of layers of BERT-large from 24 to 32. The results show that CLAN improves the throughput by 30.9\%, 56.1\%, and 67.7\% for BERT-base, BERT-large, and 
BERT-large (32 layers), respectively.

\begin{table}[h]
    \centering
    \caption{Measured throughput of BytePS-Compress with different BERT model scales.}
    \vskip -0.05in
    \begin{tabular}{|l|c|c|c|}
        \hline Model & \# Parameters & \multicolumn{2}{c|}{Throughput (seq/s)} \\
        \cline{3-4} & & LANS & CLAN \\
        \hline BERT-Base & 110M & 4613 & \textsb{6038} \\
        \hline BERT-Large & 336M & 613 & \textsb{957} \\
        \hline BERT-Large (32 layers)  & 437M & 31 & \textsb{52} \\
        \hline
    \end{tabular}
    \label{tbl:bert-throughput}
    \vskip -0.15in
\end{table}

\subsection{Ablation Study for System Optimization}

In this section, we perform an ablation experiment over several facets of BytePS-Compress as discussed in Section~\ref{sec:rco} in order to understand their relative importance better. We evaluate the throughput improvement of training the BERT-Large model with top-$k$ by incrementally adding each system optimization technique. The results are shown in Table~\ref{tbl:system-ablation}. Without any system optimization, gradient compression slows down the training by 71.78\%. Overall, the proposed system optimization techniques improve the throughput by 56.12\% over the mixed-precision method. In particular, parallelism gives the most significant improvement.

\begin{table}[h]
    \centering
    \caption{Measured throughput of BytePS-Compress with different system optimization methods.}
    \vskip -0.05in
    \begin{tabular}{|l|c|c|}
        \hline Method & Throughput (seq/s) & Speedup \\
        \hline no compression & 613 & 0\%  \\
        \hline compression w/o optimization & 173 & -71.78\% \\
        \hline + Parallelism  & 443 & -27.73\% \\
        \hline + Operator Fusion  & 499 & -18.60\% \\
        \hline + Size Threshold  & 520 & -15.17\% \\
        \hline + Workload Balance  & 796 & 29.85\% \\
        \hline + More Servers  & 909 & 48.29\% \\
        \hline + NUMA Tuning  & 957 & 56.12\% \\
        \hline
    \end{tabular}
    \label{tbl:system-ablation}
    \vskip -0.1in
\end{table}


\section{Conclusion}

In this paper, we introduce an adaptive gradient method with gradient compression called CLAN. A convergence rate of $\mathcal{O}(1/\sqrt{T})$ for non-convex problems is established. Moreover, we develop a scalable system called BytePS-Compress for gradient compression. The proposed system uses a two-step compression approach that uses different intra- and inter-node gradient compression strategies, which are aligned with the capabilities/characteristics of the underlying architectures. A high degree of parallelism is achieved by pipelining the compression and decompression on CPUs. The experiment results show that we reduce the distributed training time significantly for ResNet50, VGG16, and BERT.

\bibliographystyle{ACM-Reference-Format}
\bibliography{sample-base}


\begin{thebibliography}{41}


\ifx \showCODEN    \undefined \def \showCODEN     #1{\unskip}     \fi
\ifx \showDOI      \undefined \def \showDOI       #1{#1}\fi
\ifx \showISBNx    \undefined \def \showISBNx     #1{\unskip}     \fi
\ifx \showISBNxiii \undefined \def \showISBNxiii  #1{\unskip}     \fi
\ifx \showISSN     \undefined \def \showISSN      #1{\unskip}     \fi
\ifx \showLCCN     \undefined \def \showLCCN      #1{\unskip}     \fi
\ifx \shownote     \undefined \def \shownote      #1{#1}          \fi
\ifx \showarticletitle \undefined \def \showarticletitle #1{#1}   \fi
\ifx \showURL      \undefined \def \showURL       {\relax}        \fi
\providecommand\bibfield[2]{#2}
\providecommand\bibinfo[2]{#2}
\providecommand\natexlab[1]{#1}
\providecommand\showeprint[2][]{arXiv:#2}

\bibitem[\protect\citeauthoryear{Ahmed, Elzanaty, Alouini, and Canini}{Ahmed
  et~al\mbox{.}}{2021}]%
        {m2021efficient}
\bibfield{author}{\bibinfo{person}{A. Ahmed}, \bibinfo{person}{A. Elzanaty},
  \bibinfo{person}{M.-S. Alouini}, {and} \bibinfo{person}{M. Canini}.}
  \bibinfo{year}{2021}\natexlab{}.
\newblock \showarticletitle{{An Efficient Statistical-based Gradient
  Compression Technique for Distributed Training Systems}}.
\newblock \bibinfo{journal}{\emph{Proceedings of Machine Learning and
  Systems}}.
\newblock


\bibitem[\protect\citeauthoryear{Aji and Heafield}{Aji and Heafield}{2017}]%
        {aji2017sparse}
\bibfield{author}{\bibinfo{person}{A. Aji} {and} \bibinfo{person}{K.
  Heafield}.} \bibinfo{year}{2017}\natexlab{}.
\newblock \showarticletitle{{Sparse Communication for Distributed Gradient
  Descent}}. In \bibinfo{booktitle}{\emph{Proceedings of Conference on
  Empirical Methods in Natural Language Processing}}.
\newblock


\bibitem[\protect\citeauthoryear{Alistarh, Grubic, Li, Tomioka, and
  Vojnovic}{Alistarh et~al\mbox{.}}{2017}]%
        {alistarh2017qsgd}
\bibfield{author}{\bibinfo{person}{D. Alistarh}, \bibinfo{person}{D. Grubic},
  \bibinfo{person}{J. Li}, \bibinfo{person}{R. Tomioka}, {and}
  \bibinfo{person}{M. Vojnovic}.} \bibinfo{year}{2017}\natexlab{}.
\newblock \showarticletitle{{{QSGD:} Communication-Efficient {SGD} via Gradient
  Quantization and Encoding}}. In \bibinfo{booktitle}{\emph{Proceedings of
  Advances in Neural Information Processing Systems}}.
\newblock


\bibitem[\protect\citeauthoryear{Basu, Data, Karakus, and Diggavi}{Basu
  et~al\mbox{.}}{2019}]%
        {basu2019qsparse}
\bibfield{author}{\bibinfo{person}{D. Basu}, \bibinfo{person}{D. Data},
  \bibinfo{person}{C. Karakus}, {and} \bibinfo{person}{S.~N. Diggavi}.}
  \bibinfo{year}{2019}\natexlab{}.
\newblock \showarticletitle{{Qsparse-local-SGD: Distributed {SGD} with
  Quantization, Sparsification and Local Computations}}. In
  \bibinfo{booktitle}{\emph{Proceedings of Advances in Neural Information
  Processing Systems}}.
\newblock


\bibitem[\protect\citeauthoryear{Bernstein, Wang, Azizzadenesheli, and
  Anandkumar}{Bernstein et~al\mbox{.}}{2018}]%
        {bernstein2018signsgd}
\bibfield{author}{\bibinfo{person}{J. Bernstein}, \bibinfo{person}{Y. Wang},
  \bibinfo{person}{K. Azizzadenesheli}, {and} \bibinfo{person}{A. Anandkumar}.}
  \bibinfo{year}{2018}\natexlab{}.
\newblock \showarticletitle{{{SIGNSGD:} Compressed Optimisation for Non-Convex
  Problems}}. In \bibinfo{booktitle}{\emph{Proceedings of International
  Conference on Machine Learning}}.
\newblock


\bibitem[\protect\citeauthoryear{Brown, Mann, Ryder, Subbiah, Kaplan, Dhariwal,
  Neelakantan, Shyam, Sastry, et~al\mbox{.}}{Brown et~al\mbox{.}}{2020}]%
        {brown2020language}
\bibfield{author}{\bibinfo{person}{T.~B. Brown}, \bibinfo{person}{B. Mann},
  \bibinfo{person}{N. Ryder}, \bibinfo{person}{M. Subbiah}, \bibinfo{person}{J.
  Kaplan}, \bibinfo{person}{P. Dhariwal}, \bibinfo{person}{A. Neelakantan},
  \bibinfo{person}{P. Shyam}, \bibinfo{person}{G. Sastry}, {et~al\mbox{.}}}
  \bibinfo{year}{2020}\natexlab{}.
\newblock \showarticletitle{{Language Models are Few-Shot Learners}}. In
  \bibinfo{booktitle}{\emph{Proceedings of Advances in Neural Information
  Processing Systems}}.
\newblock


\bibitem[\protect\citeauthoryear{Chen, Shen, Huang, Liu, and Luo}{Chen
  et~al\mbox{.}}{2020}]%
        {chen2020efficient}
\bibfield{author}{\bibinfo{person}{C. Chen}, \bibinfo{person}{L. Shen},
  \bibinfo{person}{H. Huang}, \bibinfo{person}{W. Liu}, {and}
  \bibinfo{person}{Z.-Q. Luo}.} \bibinfo{year}{2020}\natexlab{}.
\newblock \showarticletitle{{Efficient-Adam: Communication-Efficient
  Distributed Adam with Complexity Analysis}}.
\newblock  (\bibinfo{year}{2020}).
\newblock


\bibitem[\protect\citeauthoryear{Darken and Moody}{Darken and Moody}{1991}]%
        {darken1991towards}
\bibfield{author}{\bibinfo{person}{C. Darken} {and} \bibinfo{person}{J.
  Moody}.} \bibinfo{year}{1991}\natexlab{}.
\newblock \showarticletitle{{Towards Faster Stochastic Gradient Search}}. In
  \bibinfo{booktitle}{\emph{Proceedings of Advances in Neural Information
  Processing Systems}}.
\newblock


\bibitem[\protect\citeauthoryear{Dean, Corrado, Monga, Chen, Devin, Le, Mao,
  Ranzato, Senior, Tucker, Yang, and Ng}{Dean et~al\mbox{.}}{2012}]%
        {dean2012large}
\bibfield{author}{\bibinfo{person}{J. Dean}, \bibinfo{person}{G. Corrado},
  \bibinfo{person}{R. Monga}, \bibinfo{person}{K. Chen}, \bibinfo{person}{M.
  Devin}, \bibinfo{person}{Q.~V. Le}, \bibinfo{person}{M.~Z. Mao},
  \bibinfo{person}{M. Ranzato}, \bibinfo{person}{A.~W. Senior},
  \bibinfo{person}{P.~A. Tucker}, \bibinfo{person}{K. Yang}, {and}
  \bibinfo{person}{A.~Y. Ng}.} \bibinfo{year}{2012}\natexlab{}.
\newblock \showarticletitle{{Large Scale Distributed Deep Networks}}. In
  \bibinfo{booktitle}{\emph{Proceedings of Advances in Neural Information
  Processing Systems}}.
\newblock


\bibitem[\protect\citeauthoryear{Devlin, Chang, Lee, and Toutanova}{Devlin
  et~al\mbox{.}}{2019}]%
        {devlin2018bert}
\bibfield{author}{\bibinfo{person}{J. Devlin}, \bibinfo{person}{M. Chang},
  \bibinfo{person}{K. Lee}, {and} \bibinfo{person}{K. Toutanova}.}
  \bibinfo{year}{2019}\natexlab{}.
\newblock \showarticletitle{{BERT}: Pre-training of Deep Bidirectional
  Transformers for Language Understanding}. In
  \bibinfo{booktitle}{\emph{Proceedings of the North American Chapter of the
  Association for Computational Linguistics: Human Language Technologies}}.
\newblock


\bibitem[\protect\citeauthoryear{Goyal, Doll{\'a}r, Girshick, Noordhuis,
  Wesolowski, Kyrola, Tulloch, Jia, and He}{Goyal et~al\mbox{.}}{2017}]%
        {goyal2017accurate}
\bibfield{author}{\bibinfo{person}{P. Goyal}, \bibinfo{person}{P. Doll{\'a}r},
  \bibinfo{person}{R. Girshick}, \bibinfo{person}{P. Noordhuis},
  \bibinfo{person}{L. Wesolowski}, \bibinfo{person}{A. Kyrola},
  \bibinfo{person}{A. Tulloch}, \bibinfo{person}{Y. Jia}, {and}
  \bibinfo{person}{K. He}.} \bibinfo{year}{2017}\natexlab{}.
\newblock \showarticletitle{{Accurate, Large Minibatch SGD: Training ImageNet
  in 1 Hour}}.
\newblock \bibinfo{journal}{\emph{arXiv preprint}} (\bibinfo{year}{2017}).
\newblock


\bibitem[\protect\citeauthoryear{Guo, He, He, Lausen, Li, Lin, Shi, Wang, Xie,
  Zha, et~al\mbox{.}}{Guo et~al\mbox{.}}{2020}]%
        {guo2020gluoncv}
\bibfield{author}{\bibinfo{person}{J. Guo}, \bibinfo{person}{H. He},
  \bibinfo{person}{T. He}, \bibinfo{person}{L. Lausen}, \bibinfo{person}{M.
  Li}, \bibinfo{person}{H. Lin}, \bibinfo{person}{X. Shi}, \bibinfo{person}{C.
  Wang}, \bibinfo{person}{J. Xie}, \bibinfo{person}{S. Zha}, {et~al\mbox{.}}}
  \bibinfo{year}{2020}\natexlab{}.
\newblock \showarticletitle{{GluonCV and GluonNLP: Deep Learning in Computer
  Vision and Natural Language Processing}}.
\newblock \bibinfo{journal}{\emph{Journal of Machine Learning Research}}
  (\bibinfo{year}{2020}).
\newblock


\bibitem[\protect\citeauthoryear{Gupta, Choudhary, Tang, Wei, Wang, Huang,
  Kejariwal, Ramchandran, and Mahoney}{Gupta et~al\mbox{.}}{2020}]%
        {gupta2020fast}
\bibfield{author}{\bibinfo{person}{V. Gupta}, \bibinfo{person}{D. Choudhary},
  \bibinfo{person}{P. Tang}, \bibinfo{person}{X. Wei}, \bibinfo{person}{X.
  Wang}, \bibinfo{person}{Y. Huang}, \bibinfo{person}{A. Kejariwal},
  \bibinfo{person}{K. Ramchandran}, {and} \bibinfo{person}{M. Mahoney}.}
  \bibinfo{year}{2020}\natexlab{}.
\newblock \showarticletitle{{Fast Distributed Training of Deep Neural Networks:
  Dynamic Communication Thresholding for Model and Data Parallelism}}.
\newblock \bibinfo{journal}{\emph{arXiv preprint}} (\bibinfo{year}{2020}).
\newblock


\bibitem[\protect\citeauthoryear{H, Gan, Rajbhandari, Lian, Zhang, Liu, and
  He}{H et~al\mbox{.}}{2020}]%
        {tang2020apmsqueeze}
\bibfield{author}{\bibinfo{person}{Tang H}, \bibinfo{person}{S. Gan},
  \bibinfo{person}{S. Rajbhandari}, \bibinfo{person}{X. Lian},
  \bibinfo{person}{C. Zhang}, \bibinfo{person}{J. Liu}, {and}
  \bibinfo{person}{Y. He}.} \bibinfo{year}{2020}\natexlab{}.
\newblock \showarticletitle{{APMSqueeze: A Communication Efficient
  Adam-Preconditioned Momentum SGD Algorithm}}.
\newblock \bibinfo{journal}{\emph{arXiv preprint}} (\bibinfo{year}{2020}).
\newblock


\bibitem[\protect\citeauthoryear{He, Zhang, Ren, and Sun}{He
  et~al\mbox{.}}{2016}]%
        {he2016identity}
\bibfield{author}{\bibinfo{person}{K. He}, \bibinfo{person}{X. Zhang},
  \bibinfo{person}{S. Ren}, {and} \bibinfo{person}{J. Sun}.}
  \bibinfo{year}{2016}\natexlab{}.
\newblock \showarticletitle{{Identity Mappings in Deep Residual Networks}}. In
  \bibinfo{booktitle}{\emph{Proceedings of European Conference on Computer
  Vision}}.
\newblock


\bibitem[\protect\citeauthoryear{Horvath, Ho, Horvath, Sahu, Canini, and
  Richtarik}{Horvath et~al\mbox{.}}{2019}]%
        {horvath2019natural}
\bibfield{author}{\bibinfo{person}{S. Horvath}, \bibinfo{person}{C.-Y. Ho},
  \bibinfo{person}{L. Horvath}, \bibinfo{person}{A. Sahu}, \bibinfo{person}{M.
  Canini}, {and} \bibinfo{person}{P. Richtarik}.}
  \bibinfo{year}{2019}\natexlab{}.
\newblock \showarticletitle{{Natural Compression for Distributed Deep
  Learning}}.
\newblock \bibinfo{journal}{\emph{arXiv preprint}} (\bibinfo{year}{2019}).
\newblock


\bibitem[\protect\citeauthoryear{Horv{\'a}th and Richt{\'a}rik}{Horv{\'a}th and
  Richt{\'a}rik}{2021}]%
        {horvath2020better}
\bibfield{author}{\bibinfo{person}{S. Horv{\'a}th} {and} \bibinfo{person}{P.
  Richt{\'a}rik}.} \bibinfo{year}{2021}\natexlab{}.
\newblock \showarticletitle{{A Better Alternative to Error Feedback for
  Communication-Efficient Distributed Learning}}. In
  \bibinfo{booktitle}{\emph{Proceedings of International Conference on Learning
  Representations}}.
\newblock


\bibitem[\protect\citeauthoryear{Jiang, Zhu, Lan, Yi, Cui, and Guo}{Jiang
  et~al\mbox{.}}{2020}]%
        {jiang2020unified}
\bibfield{author}{\bibinfo{person}{Y. Jiang}, \bibinfo{person}{Y. Zhu},
  \bibinfo{person}{C. Lan}, \bibinfo{person}{B. Yi}, \bibinfo{person}{Y. Cui},
  {and} \bibinfo{person}{C. Guo}.} \bibinfo{year}{2020}\natexlab{}.
\newblock \showarticletitle{{A Unified Architecture for Accelerating
  Distributed DNN Training in Heterogeneous GPU/CPU Clusters}}. In
  \bibinfo{booktitle}{\emph{Proceedings of USENIX Symposium on Operating
  Systems Design and Implementation}}.
\newblock


\bibitem[\protect\citeauthoryear{Kaplan, McCandlish, Henighan, Brown, Chess,
  et~al\mbox{.}}{Kaplan et~al\mbox{.}}{2020}]%
        {kaplan2020scaling}
\bibfield{author}{\bibinfo{person}{J. Kaplan}, \bibinfo{person}{S. McCandlish},
  \bibinfo{person}{T. Henighan}, \bibinfo{person}{T. Brown},
  \bibinfo{person}{B. Chess}, {et~al\mbox{.}}} \bibinfo{year}{2020}\natexlab{}.
\newblock \showarticletitle{{Scaling Laws for Neural Language Models}}.
\newblock \bibinfo{journal}{\emph{arXiv preprint}} (\bibinfo{year}{2020}).
\newblock


\bibitem[\protect\citeauthoryear{Karimireddy, Rebjock, Stich, and
  Jaggi}{Karimireddy et~al\mbox{.}}{2019}]%
        {karimireddy2019error}
\bibfield{author}{\bibinfo{person}{S.~P. Karimireddy}, \bibinfo{person}{Q.
  Rebjock}, \bibinfo{person}{S.~U. Stich}, {and} \bibinfo{person}{M. Jaggi}.}
  \bibinfo{year}{2019}\natexlab{}.
\newblock \showarticletitle{{Error Feedback Fixes SignSGD and other Gradient
  Compression Schemes}}. In \bibinfo{booktitle}{\emph{Proceedings of
  International Conference on Machine Learning}}.
\newblock


\bibitem[\protect\citeauthoryear{Kingma and Ba}{Kingma and Ba}{2015}]%
        {kingma2014adam}
\bibfield{author}{\bibinfo{person}{D.~P. Kingma} {and} \bibinfo{person}{J.
  Ba}.} \bibinfo{year}{2015}\natexlab{}.
\newblock \showarticletitle{{Adam: {A} Method for Stochastic Optimization}}. In
  \bibinfo{booktitle}{\emph{Proceedings of International Conference on Learning
  Representations}}.
\newblock


\bibitem[\protect\citeauthoryear{Li, Andersen, Park, Smola, Ahmed, Josifovski,
  Long, Shekita, and Su}{Li et~al\mbox{.}}{2014}]%
        {li2014scaling}
\bibfield{author}{\bibinfo{person}{M. Li}, \bibinfo{person}{D. Andersen},
  \bibinfo{person}{J. Park}, \bibinfo{person}{A. Smola}, \bibinfo{person}{A.
  Ahmed}, \bibinfo{person}{V. Josifovski}, \bibinfo{person}{J. Long},
  \bibinfo{person}{E. Shekita}, {and} \bibinfo{person}{B.-Y. Su}.}
  \bibinfo{year}{2014}\natexlab{}.
\newblock \showarticletitle{{Scaling Distributed Machine Learning with the
  Parameter Server}}. In \bibinfo{booktitle}{\emph{Proceedings of USENIX
  Symposium on Operating Systems Design and Implementation}}.
\newblock


\bibitem[\protect\citeauthoryear{Lin, Han, Mao, Wang, and Dally}{Lin
  et~al\mbox{.}}{2018}]%
        {lin2017deep}
\bibfield{author}{\bibinfo{person}{Y. Lin}, \bibinfo{person}{S. Han},
  \bibinfo{person}{H. Mao}, \bibinfo{person}{Y. Wang}, {and}
  \bibinfo{person}{B. Dally}.} \bibinfo{year}{2018}\natexlab{}.
\newblock \showarticletitle{{Deep Gradient Compression: Reducing the
  Communication Bandwidth for Distributed Training}}. In
  \bibinfo{booktitle}{\emph{Proceedings of International Conference on Learning
  Representations}}.
\newblock


\bibitem[\protect\citeauthoryear{Narayanan, Harlap, Phanishayee, Seshadri,
  Devanur, Ganger, Gibbons, and Zaharia}{Narayanan et~al\mbox{.}}{2019}]%
        {narayanan2019pipedream}
\bibfield{author}{\bibinfo{person}{D. Narayanan}, \bibinfo{person}{A. Harlap},
  \bibinfo{person}{A. Phanishayee}, \bibinfo{person}{V. Seshadri},
  \bibinfo{person}{N. Devanur}, \bibinfo{person}{G. Ganger},
  \bibinfo{person}{P. Gibbons}, {and} \bibinfo{person}{M. Zaharia}.}
  \bibinfo{year}{2019}\natexlab{}.
\newblock \showarticletitle{{PipeDream: Generalized Pipeline Parallelism for
  DNN Training}}. In \bibinfo{booktitle}{\emph{Proceedings of ACM Symposium on
  Operating Systems Principles}}.
\newblock


\bibitem[\protect\citeauthoryear{Peng, Zhu, Chen, Bao, Yi, Lan, Wu, and
  Guo}{Peng et~al\mbox{.}}{2019}]%
        {peng2019generic}
\bibfield{author}{\bibinfo{person}{Y. Peng}, \bibinfo{person}{Y. Zhu},
  \bibinfo{person}{Y. Chen}, \bibinfo{person}{Y. Bao}, \bibinfo{person}{B. Yi},
  \bibinfo{person}{C. Lan}, \bibinfo{person}{C. Wu}, {and} \bibinfo{person}{C.
  Guo}.} \bibinfo{year}{2019}\natexlab{}.
\newblock \showarticletitle{{A Generic Communication Scheduler for Distributed
  DNN Training Acceleration}}. In \bibinfo{booktitle}{\emph{Proceedings of ACM
  Symposium on Operating Systems Principles}}.
\newblock


\bibitem[\protect\citeauthoryear{Reddi, Kale, and Kumar}{Reddi
  et~al\mbox{.}}{2018}]%
        {reddi2019convergence}
\bibfield{author}{\bibinfo{person}{S.~J. Reddi}, \bibinfo{person}{S. Kale},
  {and} \bibinfo{person}{S. Kumar}.} \bibinfo{year}{2018}\natexlab{}.
\newblock \showarticletitle{{On the Convergence of Adam and Beyond}}. In
  \bibinfo{booktitle}{\emph{Proceedings of International Conference on Learning
  Representations}}.
\newblock


\bibitem[\protect\citeauthoryear{Russakovsky, Deng, Su, Krause, Satheesh, Ma,
  Huang, Karpathy, Khosla, Bernstein, et~al\mbox{.}}{Russakovsky
  et~al\mbox{.}}{2015}]%
        {russakovsky2015imagenet}
\bibfield{author}{\bibinfo{person}{O. Russakovsky}, \bibinfo{person}{J. Deng},
  \bibinfo{person}{H. Su}, \bibinfo{person}{J. Krause}, \bibinfo{person}{S.
  Satheesh}, \bibinfo{person}{S. Ma}, \bibinfo{person}{Z. Huang},
  \bibinfo{person}{A. Karpathy}, \bibinfo{person}{A. Khosla},
  \bibinfo{person}{M. Bernstein}, {et~al\mbox{.}}}
  \bibinfo{year}{2015}\natexlab{}.
\newblock \showarticletitle{{Imagenet Large Scale Visual Recognition
  Challenge}}.
\newblock \bibinfo{journal}{\emph{Proceedings of International Journal of
  Computer Vision}}.
\newblock


\bibitem[\protect\citeauthoryear{Seide, Fu, Droppo, Li, and Yu}{Seide
  et~al\mbox{.}}{2014}]%
        {seide20141}
\bibfield{author}{\bibinfo{person}{F. Seide}, \bibinfo{person}{H. Fu},
  \bibinfo{person}{J. Droppo}, \bibinfo{person}{G. Li}, {and}
  \bibinfo{person}{D. Yu}.} \bibinfo{year}{2014}\natexlab{}.
\newblock \showarticletitle{{1-Bit Stochastic Gradient Descent and its
  Application to Data-Parallel Distributed Training of Speech DNNs}}. In
  \bibinfo{booktitle}{\emph{Proceedings of the International Speech
  Communication Association}}.
\newblock


\bibitem[\protect\citeauthoryear{Sergeev and Balso}{Sergeev and Balso}{2018}]%
        {sergeev2018horovod}
\bibfield{author}{\bibinfo{person}{A. Sergeev} {and} \bibinfo{person}{M.~Del
  Balso}.} \bibinfo{year}{2018}\natexlab{}.
\newblock \showarticletitle{{Horovod: Fast and Easy Distributed Deep Learning
  in TensorFlow}}.
\newblock \bibinfo{journal}{\emph{arXiv preprint}} (\bibinfo{year}{2018}).
\newblock


\bibitem[\protect\citeauthoryear{Simonyan and Zisserman}{Simonyan and
  Zisserman}{2015}]%
        {simonyan2014very}
\bibfield{author}{\bibinfo{person}{K. Simonyan} {and} \bibinfo{person}{A.
  Zisserman}.} \bibinfo{year}{2015}\natexlab{}.
\newblock \showarticletitle{{Very Deep Convolutional Networks for Large-Scale
  Image Recognition}}. In \bibinfo{booktitle}{\emph{Proceedings of
  International Conference on Learning Representations}}.
\newblock


\bibitem[\protect\citeauthoryear{Stich, Cordonnier, and Jaggi}{Stich
  et~al\mbox{.}}{2018}]%
        {stich2018sparsified}
\bibfield{author}{\bibinfo{person}{S.~U. Stich}, \bibinfo{person}{J.
  Cordonnier}, {and} \bibinfo{person}{M. Jaggi}.}
  \bibinfo{year}{2018}\natexlab{}.
\newblock \showarticletitle{{Sparsified {SGD} with Memory}}. In
  \bibinfo{booktitle}{\emph{Proceedings of Advances in Neural Information
  Processing Systems}}.
\newblock


\bibitem[\protect\citeauthoryear{Sutskever, Martens, Dahl, and
  Hinton}{Sutskever et~al\mbox{.}}{2013}]%
        {sutskever2013importance}
\bibfield{author}{\bibinfo{person}{I. Sutskever}, \bibinfo{person}{J. Martens},
  \bibinfo{person}{G.~E. Dahl}, {and} \bibinfo{person}{G.~E. Hinton}.}
  \bibinfo{year}{2013}\natexlab{}.
\newblock \showarticletitle{{On the Importance of Initialization and Momentum
  in Deep Learning}}. In \bibinfo{booktitle}{\emph{Proceedings of International
  Conference on Machine Learning}}.
\newblock


\bibitem[\protect\citeauthoryear{Wangni, Wang, Liu, and Zhang}{Wangni
  et~al\mbox{.}}{2018}]%
        {wangni2018gradient}
\bibfield{author}{\bibinfo{person}{J. Wangni}, \bibinfo{person}{J. Wang},
  \bibinfo{person}{K. Liu}, {and} \bibinfo{person}{T. Zhang}.}
  \bibinfo{year}{2018}\natexlab{}.
\newblock \showarticletitle{{Gradient Sparsification for
  Communication-Efficient Distributed Optimization}}. In
  \bibinfo{booktitle}{\emph{Proceedings of Advances in Neural Information
  Processing Systems}}.
\newblock


\bibitem[\protect\citeauthoryear{Wen, Xu, Yan, Wu, Wang, Chen, and Li}{Wen
  et~al\mbox{.}}{2017}]%
        {wen2017terngrad}
\bibfield{author}{\bibinfo{person}{W. Wen}, \bibinfo{person}{C. Xu},
  \bibinfo{person}{F. Yan}, \bibinfo{person}{C. Wu}, \bibinfo{person}{Y. Wang},
  \bibinfo{person}{Y. Chen}, {and} \bibinfo{person}{H. Li}.}
  \bibinfo{year}{2017}\natexlab{}.
\newblock \showarticletitle{{TernGrad: Ternary Gradients to Reduce
  Communication in Distributed Deep Learning}}. In
  \bibinfo{booktitle}{\emph{Proceedings of Advances in Neural Information
  Processing Systems}}.
\newblock


\bibitem[\protect\citeauthoryear{Xie, Zheng, Koyejo, Gupta, Li, and Lin}{Xie
  et~al\mbox{.}}{2020}]%
        {xie2020cser}
\bibfield{author}{\bibinfo{person}{C. Xie}, \bibinfo{person}{S. Zheng},
  \bibinfo{person}{O. Koyejo}, \bibinfo{person}{I. Gupta}, \bibinfo{person}{M.
  Li}, {and} \bibinfo{person}{H. Lin}.} \bibinfo{year}{2020}\natexlab{}.
\newblock \showarticletitle{{CSER: Communication-efficient SGD with Error
  Reset}}. In \bibinfo{booktitle}{\emph{Proceedings of Advances in Neural
  Information Processing Systems}}.
\newblock


\bibitem[\protect\citeauthoryear{Xu, Ho, Abdelmoniem, Dutta, Bergou,
  Karatsenidis, Canini, and Kalnis}{Xu et~al\mbox{.}}{2020}]%
        {xu2020compressed}
\bibfield{author}{\bibinfo{person}{H. Xu}, \bibinfo{person}{C.-Y. Ho},
  \bibinfo{person}{A. Abdelmoniem}, \bibinfo{person}{A. Dutta},
  \bibinfo{person}{E. Bergou}, \bibinfo{person}{K. Karatsenidis},
  \bibinfo{person}{M. Canini}, {and} \bibinfo{person}{P. Kalnis}.}
  \bibinfo{year}{2020}\natexlab{}.
\newblock \bibinfo{booktitle}{\emph{{Compressed Communication for Distributed
  Deep Learning: Survey and Quantitative Evaluation}}}.
\newblock \bibinfo{type}{{T}echnical {R}eport}.
\newblock


\bibitem[\protect\citeauthoryear{You, Li, Reddi, Hseu, Kumar, Bhojanapalli,
  Song, et~al\mbox{.}}{You et~al\mbox{.}}{2020}]%
        {you2019large}
\bibfield{author}{\bibinfo{person}{Y. You}, \bibinfo{person}{J. Li},
  \bibinfo{person}{S.~J. Reddi}, \bibinfo{person}{J. Hseu}, \bibinfo{person}{S.
  Kumar}, \bibinfo{person}{S. Bhojanapalli}, \bibinfo{person}{X. Song},
  {et~al\mbox{.}}} \bibinfo{year}{2020}\natexlab{}.
\newblock \showarticletitle{{Large Batch Optimization for Deep Learning:
  Training {BERT} in 76 Minutes}}. In \bibinfo{booktitle}{\emph{Proceedings of
  International Conference on Learning Representations}}.
\newblock


\bibitem[\protect\citeauthoryear{Zaheer, Reddi, Sachan, Kale, and Kumar}{Zaheer
  et~al\mbox{.}}{2018}]%
        {zaheer2018adaptive}
\bibfield{author}{\bibinfo{person}{M. Zaheer}, \bibinfo{person}{S.~J. Reddi},
  \bibinfo{person}{D.~S. Sachan}, \bibinfo{person}{S. Kale}, {and}
  \bibinfo{person}{S. Kumar}.} \bibinfo{year}{2018}\natexlab{}.
\newblock \showarticletitle{{Adaptive Methods for Nonconvex Optimization}}. In
  \bibinfo{booktitle}{\emph{Proceedings of Advances in Neural Information
  Processing Systems}}.
\newblock


\bibitem[\protect\citeauthoryear{Zhang, Karimireddy, Veit, Kim, Reddi, Kumar,
  and Sra}{Zhang et~al\mbox{.}}{2020}]%
        {zhang2020adaptive}
\bibfield{author}{\bibinfo{person}{J. Zhang}, \bibinfo{person}{S. Karimireddy},
  \bibinfo{person}{A. Veit}, \bibinfo{person}{S. Kim}, \bibinfo{person}{S.
  Reddi}, \bibinfo{person}{S. Kumar}, {and} \bibinfo{person}{S. Sra}.}
  \bibinfo{year}{2020}\natexlab{}.
\newblock \showarticletitle{{Why are Adaptive Methods Good for Attention
  Models?}}
\newblock \bibinfo{journal}{\emph{Proceedings of Advances in Neural Information
  Processing Systems}}.
\newblock


\bibitem[\protect\citeauthoryear{Zheng, Huang, and Kwok}{Zheng
  et~al\mbox{.}}{2019}]%
        {zheng2019communication}
\bibfield{author}{\bibinfo{person}{S. Zheng}, \bibinfo{person}{Z. Huang}, {and}
  \bibinfo{person}{J.~T. Kwok}.} \bibinfo{year}{2019}\natexlab{}.
\newblock \showarticletitle{{Communication-Efficient Distributed Blockwise
  Momentum {SGD} with Error-Feedback}}. In
  \bibinfo{booktitle}{\emph{Proceedings of Advances in Neural Information
  Processing Systems}}.
\newblock


\bibitem[\protect\citeauthoryear{Zheng, Lin, Zha, and Li}{Zheng
  et~al\mbox{.}}{2020}]%
        {zheng2020accelerated}
\bibfield{author}{\bibinfo{person}{S. Zheng}, \bibinfo{person}{H. Lin},
  \bibinfo{person}{S. Zha}, {and} \bibinfo{person}{M. Li}.}
  \bibinfo{year}{2020}\natexlab{}.
\newblock \showarticletitle{{Accelerated Large Batch Optimization of BERT
  Pretraining in 54 Minutes}}.
\newblock \bibinfo{journal}{\emph{arXiv preprint}} (\bibinfo{year}{2020}).
\newblock


\end{thebibliography}

\appendix

\newpage
\onecolumn

\allowdisplaybreaks

\section{Proof}

\subsection{Convergence Analysis}

In the following theoretical analysis, we consider LANS using a general gradient estimator $p_t$ (the block index is ignored) in the $t^{\mbox{th}}$ iteration. For LANS with full precision in pushpull operations, $p_t = \frac{1}{n}\sum_{i \in [n]} g_{t, i}$. For LANS with unbiased compressors, we have $p_t = \mathcal{C}\left( \frac{1}{n} \sum_{i \in [n]} \mathcal{C}(g_{t, i}) \right)$, where $\E[p_t] = \frac{1}{n}\sum_{i \in [n]} g_{t, i}$. For LANS with biased compressors and error feedback, we have $p_t =  \mathcal{C}\left( \frac{1}{n} \sum_{i \in [n]} \mathcal{C}(g_{t, i} + e_{t, i}) + \tilde{e}_t \right)$.

For simplicity,  throughout the theoretical analysis, we ignore the regularization, i.e. taking $\lambda = 0$, and assume fixed learning rate $\eta_t = \eta, \forall t \in [T]$.






\begin{lemma}~\label{lem:mom_err}
For any block $b \in [B]$, $j \in \mG_b$, and $t \in [T]$ in LANS with the general gradient estimator $p_t$, we have the upper bound of the difference between the first moment (with bias correction) and the gradient:
\begin{align*}
&| (\tm_{t, \mG_b} - \nabla_{\mG_b} F(x_{t}))_j | 
\leq \frac{\beta_1}{1-\beta_1} L_{b,j} \eta \alpha_u
+ | (p_{t, \mG_b} - \nabla_{\mG_b} F(x_{t}))_j |
+ \sum_{\tau = 1}^{t-1}  \frac{\beta_1^\tau}{1-\beta_1^t} | (p_{t-\tau, \mG_b} - \nabla_{\mG_b} F(x_{t-\tau}))_j |.
\end{align*}
\end{lemma}

\begin{proof}

For LANS, it is easy to check that
\begin{align*}
\tm_{t, \mG_b} = \frac{m_{t, \mG_b}}{1 - \beta_1^t} = \frac{\beta_1 (1-\beta_1^{t-1})}{1-\beta_1^t} \tm_{t-1, \mG_b} + \frac{1-\beta_1}{1-\beta_1^t} p_{t, \mG_b}.
\end{align*}

Thus, we have
\begin{align*}
&| (\tm_{t, \mG_b} - \nabla_{\mG_b} F(x_{t}))_j | \\
&\leq \frac{\beta_1 (1-\beta_1^{t-1})}{1-\beta_1^t} | (\tm_{t-1, \mG_b} - \nabla_{\mG_b} F(x_{t}))_j | + \frac{1-\beta_1}{1-\beta_1^t} | (p_{t, \mG_b} - \nabla_{\mG_b} F(x_{t}))_j | \\
&= \frac{\beta_1 (1-\beta_1^{t-1})}{1-\beta_1^t} | (\tm_{t-1, \mG_b} - \nabla_{\mG_b} F(x_{t-1}) + \nabla_{\mG_b} F(x_{t-1}) - \nabla_{\mG_b} F(x_{t}))_j | \\
&\quad + \frac{1-\beta_1}{1-\beta_1^t} | (p_{t, \mG_b} - \nabla_{\mG_b} F(x_{t}))_j | \\
&\leq \frac{\beta_1 (1-\beta_1^{t-1})}{1-\beta_1^t} | (\tm_{t-1, \mG_b} - \nabla_{\mG_b} F(x_{t-1}))_j | 
+ \frac{\beta_1 (1-\beta_1^{t-1})}{1-\beta_1^t} | (\nabla_{\mG_b} F(x_{t-1}) - \nabla_{\mG_b} F(x_{t}))_j | \\
&\quad + \frac{1-\beta_1}{1-\beta_1^t} | (p_{t, \mG_b} - \nabla_{\mG_b} F(x_{t}))_j | \\
&\leq \frac{\beta_1 (1-\beta_1^{t-1})}{1-\beta_1^t} | (\tm_{t-1, \mG_b} - \nabla_{\mG_b} F(x_{t-1}))_j | 
+ \frac{\beta_1 (1-\beta_1^{t-1})}{1-\beta_1^t} L_{b,j} | (x_{t-1} - x_{t})_j | \comment{coordinate-wise smoothness} \\
&\quad + \frac{1-\beta_1}{1-\beta_1^t} | (p_{t, \mG_b} - \nabla_{\mG_b} F(x_{t}))_j | \\
&= \frac{\beta_1 (1-\beta_1^{t-1})}{1-\beta_1^t} | (\tm_{t-1, \mG_b} - \nabla_{\mG_b} F(x_{t-1}))_j |  + \frac{\beta_1 (1-\beta_1^{t-1})}{1-\beta_1^t} L_{b,j} \eta | ( \tilde{d}_{t-1, \mG_b})_j |  \\
&\quad + \frac{1-\beta_1}{1-\beta_1^t} | (p_{t, \mG_b} - \nabla_{\mG_b} F(x_{t}))_j | \\
&\leq \frac{\beta_1 (1-\beta_1^{t-1})}{1-\beta_1^t} | (\tm_{t-1, \mG_b} - \nabla_{\mG_b} F(x_{t-1}))_j | 
+ \frac{\beta_1 (1-\beta_1^{t-1})}{1-\beta_1^t} L_{b,j} \eta \alpha_u 
+ \frac{1-\beta_1}{1-\beta_1^t} | (p_{t, \mG_b} - \nabla_{\mG_b} F(x_{t}))_j | \\
&\leq \frac{\beta_1 (1-\beta_1^{t-1})}{1-\beta_1^t} | (\tm_{t-1, \mG_b} - \nabla_{\mG_b} F(x_{t-1}))_j | 
+ \frac{\beta_1 (1-\beta_1^{t-1})}{1-\beta_1^t} L_{b,j} \eta \alpha_u 
+ | (p_{t, \mG_b} - \nabla_{\mG_b} F(x_{t}))_j |,
\end{align*}
where in the second-to-last inequality, we use the fact that $|(\tilde{d}_{t-1, \mG_b})_j| \leq \|\tilde{d}_{t-1}\| \leq \alpha_u$. Note that $\tm_{1, \mG_b} - \nabla_{\mG_b} F(x_{1}) = p_{1, \mG_b} - \nabla_{\mG_b} F(x_{1})$. Thus, by unrolling the sequence, we have
\begin{align*}
&| (\tm_{t, \mG_b} - \nabla_{\mG_b} F(x_{t}))_j | \\
&\leq \frac{\beta_1 (1-\beta_1^{t-1})}{1-\beta_1^t} \ldots \frac{\beta_1 (1-\beta_1)}{1-\beta_1^{2}} | (p_{1, \mG_b} - \nabla_{\mG_b} F(x_{1}))_j | \\
&\quad + L_{b,j} \eta \alpha_u \sum_{\tau = 1}^{t-1} \frac{\beta_1 (1-\beta_1^{t-1})}{1-\beta_1^t} \ldots \frac{\beta_1 (1-\beta_1^{t-\tau})}{1-\beta_1^{t-\tau+1}}  \\
&\quad + | (p_{t, \mG_b} - \nabla_{\mG_b} F(x_{t}))_j |
+ \sum_{\tau = 1}^{t-2} \frac{\beta_1 (1-\beta_1^{t-1})}{1-\beta_1^t} \ldots \frac{\beta_1 (1-\beta_1^{t-\tau})}{1-\beta_1^{t-\tau+1}} | (p_{t-\tau, \mG_b} - \nabla_{\mG_b} F(x_{t-\tau}))_j | \\
&= L_{b,j} \eta \alpha_u \sum_{\tau = 1}^{t-1} \frac{\beta_1^\tau (1-\beta_1^{t-\tau})}{1-\beta_1^t} 
+ | (p_{t, \mG_b} - \nabla_{\mG_b} F(x_{t}))_j |
+ \sum_{\tau = 1}^{t-1}  \frac{\beta_1^\tau (1-\beta_1^{t-\tau})}{1-\beta_1^t} | (p_{t-\tau, \mG_b} - \nabla_{\mG_b} F(x_{t-\tau}))_j | \\
&\leq \frac{\beta_1}{1-\beta_1} L_{b,j} \eta \alpha_u
+ | (p_{t, \mG_b} - \nabla_{\mG_b} F(x_{t}))_j |
+ \sum_{\tau = 1}^{t-1}  \frac{\beta_1^\tau}{1-\beta_1^t} | (p_{t-\tau, \mG_b} - \nabla_{\mG_b} F(x_{t-\tau}))_j |.
\end{align*}

\end{proof}


\begin{theorem}
After $T$ iterations, for CLAN with the general gradient estimator $p_t$ and $(p_t)_j \leq V_1$, we have the following error bound:
\begin{eqnarray*}
\frac{\sum_{t=1}^{T} \E \| \nabla F(x_t) \|^2}{T}
& \leq & \frac{\sqrt{d} V'_1 \E[ F(x_{1}) - F(x_{*}) ]}{T \eta \alpha_l (1-\beta_1) \sqrt{1-\beta_2}}
+ \frac{\eta \sqrt{d} V'_1 \alpha_u^2 (1-\beta_1 + 2\beta_1^2) \|L\|_1 }{2 \sqrt{1-\beta_2} (1-\beta_1)^2 \alpha_l} \\
&& + \frac{\sqrt{d} V'_1 \alpha_u [(1-\beta_1)^2+\beta_1]}{\sqrt{1-\beta_2} (1-\beta_1)^2 \alpha_l} V_2 
+ \sqrt{d} G V_3,
\end{eqnarray*}
where $V'_1 = V_1 + \epsilon$, $\sum_{b \in [B]} \sum_{j \in \mG_b} \E |(\nabla_{\mG_b} F(x_t) - p_{t, \mG_b})_j| \leq V_2$, $\|\nabla F(x_t) - \E[p_{t}]\| \leq V_3$, $\forall t \in [T]$. The exact values of $V_1$, $V_2$, and $V_3$ depend on the choice of the gradient estimator $p_t$.
\end{theorem}

\begin{proof}
Using the block-wise smoothness, we have
\begin{align*}
F(x_{t+1}) 
\leq F(x_t) + \sum_{b \in [B]} \ip{\nabla_{\mG_b} F(x_t)}{x_{t+1, \mG_b} - x_{t, \mG_b}} + \sum_{b \in [B]} \frac{L_b}{2} \| x_{t+1, \mG_b} - x_{t, \mG_b} \|^2.
\end{align*}

In LANS, we have $x_{t + 1, \mG_b} - x_{t, \mG_b} =  - \eta_t \tilde{d}_{t, \mG_b}$. Thus, we have
\begin{align}
&\| x_{t+1, \mG_b} - x_{t, \mG_b} \| \nonumber \\
&= \eta \left\| \phi(\|x_{t, \mG_b}\|_2) \left[\frac{\beta_1}{\|r_{t, \mG_b}\|}r_{t, \mG_b} + \frac{1 - \beta_1}{\|c_{t, \mG_b}\|} c_{t, \mG_b}\right] \right\| \nonumber \\
&\leq \eta \phi(\|x_{t, \mG_b}\|_2) \beta_1 \left\| \frac{r_{t, \mG_b}}{\|r_{t, \mG_b}\|} \right\| + \eta \phi(\|x_{t, \mG_b}\|_2) (1 - \beta_1) \left\| \frac{c_{t, \mG_b}}{\|c_{t, \mG_b}\|} \right\| \nonumber \\
&\leq \eta \alpha_u, \label{eq:var_up}
\end{align}
where in the last inequality we use the upper bound of the scaling factor. Substituting the above bound in the smoothness inequality, we have
\begin{align*}
&F(x_{t+1}) \leq F(x_t) + \sum_{b \in [B]} \underbrace{ \ip{\nabla_{\mG_b} F(x_t)}{x_{t+1, \mG_b} - x_{t, \mG_b}} }_{\mcir{1}} + \sum_{b \in [B]} \frac{\eta^2 \alpha_u^2 L_b}{2}.
\end{align*}

Now we establish the upper bound of $\mcir{1}$.
\begin{align}
&\mcir{1} \nonumber \\
&= \ip{\nabla_{\mG_b} F(x_t)}{x_{t+1, \mG_b} - x_{t, \mG_b}} \nonumber \\
&= - \eta \phi(\|x_{t, \mG_b}\|_2) \ip{\nabla_{\mG_b} F(x_t)}{ \frac{\beta_1}{\|r_{t, \mG_b}\|}r_{t, \mG_b} + \frac{1 - \beta_1}{\|c_{t, \mG_b}\|} c_{t, \mG_b} } \nonumber \\
&= \underbrace{ - \eta \phi(\|x_{t, \mG_b}\|_2) \sum_{j \in \mG_b} \left[(\nabla_{\mG_b} F(x_t))_j \times \frac{\beta_1}{\|r_{t, \mG_b}\|} (r_{t, \mG_b})_j\right]}_{\mcir{2}} \underbrace{ - \eta \phi(\|x_{t, \mG_b}\|_2) \sum_{j \in \mG_b} \left[(\nabla_{\mG_b} F(x_t))_j \times \frac{1 - \beta_1}{\|c_{t, \mG_b}\|} (c_{t, \mG_b})_j\right]}_{\mcir{3}}. \label{eq:1_up}
\end{align}

To obtain the upper bound of $\mcir{1}$, we establish the upper bounds of $\mcir{2}$ and $\mcir{3}$ in expectation respectively.
\begin{align*}
&\mcir{2} \\
&= - \eta \phi(\|x_{t, \mG_b}\|_2) \sum_{j \in \mG_b} (\nabla_{\mG_b} F(x_t))_j \times \frac{\beta_1}{\|r_{t, \mG_b}\|} (r_{t, \mG_b})_j \\
&\leq \eta \phi(\|x_{t, \mG_b}\|_2) \sum_{j \in \mG_b} | (\nabla_{\mG_b} F(x_t))_j | \times \frac{\beta_1}{\|r_{t, \mG_b}\|} | (r_{t, \mG_b})_j | \1(sign((\nabla_{\mG_b} F(x_t))_j) \neq sign((\tm_{t, \mG_b})_j)) \\ 
&\leq \eta \beta_1 \alpha_u \sum_{j \in \mG_b} | (\nabla_{\mG_b} F(x_t))_j | \times \1(sign((\nabla_{\mG_b} F(x_t))_j) \neq sign((\tm_{t, \mG_b})_j)). \comment{$\frac{| (r_{t, \mG_b})_j |}{\|r_{t, \mG_b}\|} \leq 1$} 
\end{align*}

Taking expectation w.r.t. the random variable at iteration $t$, we have
\begin{align*}
&\E[ \mcir{2} ] \\
&\leq \eta \beta_1 \alpha_u \sum_{j \in \mG_b} | (\nabla_{\mG_b} F(x_t))_j | \times \E \left[ \1(sign((\nabla_{\mG_b} F(x_t))_j) \neq sign((\tm_{t, \mG_b})_j)) \right] \\
&= \eta \beta_1 \alpha_u \sum_{j \in \mG_b} | (\nabla_{\mG_b} F(x_t))_j | \times \P(sign((\nabla_{\mG_b} F(x_t))_j) \neq sign((\tm_{t, \mG_b})_j)).  
\end{align*}

The probability could be upper bounded as follows, using the trick from \cite{bernstein2018signsgd}:
\begin{align}
&\P(sign((\nabla_{\mG_b} F(x_t))_j) \neq sign((\tm_{t, \mG_b})_j)) \nonumber\\
&\leq \P( | (\nabla_{\mG_b} F(x_t))_j - (\tm_{t, \mG_b})_j| \geq |(\nabla_{\mG_b} F(x_t))_j| ) \nonumber\\
&\leq \frac{\E| (\nabla_{\mG_b} F(x_{t}) - \tm_{t, \mG_b})_j|}{|(\nabla_{\mG_b} F(x_t))_j|}, \label{eq:prob_up}
\end{align}
where in the last inequality we use Markov's inequality. We then can upper bound $\E| (\nabla_{\mG_b} F(x_{t}) - \tm_{t, \mG_b})_j|$ using Lemma~\ref{lem:mom_err}.

Substituting the above probability in $\E[\mcir{2}]$, we have
\begin{align}
\E[ \mcir{2} ] \leq \eta \beta_1 \alpha_u \sum_{j \in \mG_b} \E| (\nabla_{\mG_b} F(x_{t}) - \tm_{t, \mG_b})_j|.  \label{eq:2_up}
\end{align}

Note that
\begin{align*}
\|c_{t, \mG_b}\| 
= \left\| \frac{p_{t, \mG_b}}{\sqrt{\frac{\beta_2v_{t-1, \mG_b} + (1 - \beta_2)p_{t, \mG_b}^2}{1 - \beta_2^t}} + \epsilon} \right\|
\leq \left\| \frac{p_{t, \mG_b}}{\sqrt{ (1 - \beta_2)p_{t, \mG_b}^2 / (1 - \beta_2^t)}} \right\|
\leq \frac{1}{\sqrt{1 - \beta_2}} \left\| \frac{p_{t, \mG_b}}{\sqrt{ p_{t, \mG_b}^2}} \right\|
\leq \sqrt{\frac{d_b}{1 - \beta_2}},
\end{align*}
where $d_b$ is the length of the $b$th block.

Then, we establish the upper bound of $\mcir{3} = - \eta \phi(\|x_{t, \mG_b}\|_2) (1 - \beta_1) \sum_{j \in \mG_b} (\nabla_{\mG_b} F(x_t))_j \times \frac{(c_{t, \mG_b})_j}{\|c_{t, \mG_b}\|}$.

For any $j \in \mG_b$, if $sign((\nabla_{\mG_b} F(x_t))_j) = sign((p_{t, \mG_b})_j)$, we have
\begin{align*}
-  (\nabla_{\mG_b} F(x_t))_j \times \frac{(c_{t, \mG_b})_j}{\|c_{t, \mG_b}\|}
\leq -  \sqrt{\frac{1-\beta_2}{(V'_1)^2 d}}  (\nabla_{\mG_b} F(x_t))_j \times (p_{t, \mG_b})_j, 
\end{align*}
using $\|c_{t, \mG_b}\| \leq \sqrt{\frac{d_b}{1-\beta_2}}$, and $(\sqrt{\tv_{t, \mG_b}})_j + \epsilon \leq V_1 + \epsilon = V'_1$, where $(p_t)_j \leq V_1$.

For any $j \in \mG_b$, if $sign((\nabla_{\mG_b} F(x_t))_j) \neq sign((p_{t, \mG_b})_j)$, then we have
\begin{align*}
-  (\nabla_{\mG_b} F(x_t))_j \times \frac{(c_{t, \mG_b})_j}{\|c_{t, \mG_b}\|}
=  |(\nabla_{\mG_b} F(x_t))_j| \times \left| \frac{(c_{t, \mG_b})_j}{\|c_{t, \mG_b}\|} \right|.
\end{align*}

Thus, we have
\begin{align*}
&\mcir{3} \\
&= - \eta \phi(\|x_{t, \mG_b}\|_2) (1 - \beta_1) \sum_{j \in \mG_b} (\nabla_{\mG_b} F(x_t))_j \times \frac{(c_{t, \mG_b})_j}{\|c_{t, \mG_b}\|} \\
&\leq - \eta \phi(\|x_{t, \mG_b}\|_2) (1 - \beta_1) \sqrt{\frac{1-\beta_2}{(V'_1)^2 d}} \sum_{j \in \mG_b} (\nabla_{\mG_b} F(x_t))_j \times (p_{t, \mG_b})_j \1( sign((\nabla_{\mG_b} F(x_t))_j) = sign((p_{t, \mG_b})_j) ) \\
&\quad + \eta \phi(\|x_{t, \mG_b}\|_2) (1 - \beta_1) \sum_{j \in \mG_b} |(\nabla_{\mG_b} F(x_t))_j| \times \left| \frac{(c_{t, \mG_b})_j}{\|c_{t, \mG_b}\|} \right| \1( sign((\nabla_{\mG_b} F(x_t))_j) \neq sign((p_{t, \mG_b})_j) ) \\
&\leq - \eta \phi(\|x_{t, \mG_b}\|_2) (1 - \beta_1) \sqrt{\frac{1-\beta_2}{(V'_1)^2 d}} \sum_{j \in \mG_b} (\nabla_{\mG_b} F(x_t))_j \times (p_{t, \mG_b})_j \\
&\quad + \eta \phi(\|x_{t, \mG_b}\|_2) (1 - \beta_1) \sum_{j \in \mG_b} |(\nabla_{\mG_b} F(x_t))_j| \times \left| \frac{(c_{t, \mG_b})_j}{\|c_{t, \mG_b}\|} \right| \1( sign((\nabla_{\mG_b} F(x_t))_j) \neq sign((p_{t, \mG_b})_j) ).
\end{align*}

Taking expectation w.r.t. the random variable at iteration $t$, we have
\begin{align*}
&\E[\mcir{3}] \\
&\leq - \eta \alpha_l (1 - \beta_1) \sqrt{\frac{1-\beta_2}{(V'_1)^2 d}} \| \nabla_{\mG_b} F(x_t) \|^2 \\
&\quad + \eta \alpha_l (1 - \beta_1) \sqrt{\frac{1-\beta_2}{(V'_1)^2 d}}  \ip{\nabla_{\mG_b} F(x_t)}{\nabla_{\mG_b} F(x_t) - \E[p_{t, \mG_b}]} \\
&\quad + \eta \alpha_u (1 - \beta_1) \sum_{j \in \mG_b} |(\nabla_{\mG_b} F(x_t))_j| \times \P( sign((\nabla_{\mG_b} F(x_t))_j) \neq sign((p_{t, \mG_b})_j) ).
\end{align*}

Using the trick in \cite{bernstein2018signsgd}, the probability above is upper bounded:
\begin{align*}
&\P( sign((\nabla_{\mG_b} F(x_t))_j) \neq sign((p_{t, \mG_b})_j) ) \\
&\leq \P( |(\nabla_{\mG_b} F(x_t))_j - (p_{t, \mG_b})_j| \geq |(\nabla_{\mG_b} F(x_t))_j| ) \\
&\leq \frac{\E |(\nabla_{\mG_b} F(x_t) - p_{t, \mG_b})_j|}{|(\nabla_{\mG_b} F(x_t))_j|}.
\end{align*}


Thus, we have
\begin{align}
\nonumber
\E[\mcir{3}] 
&\leq - \eta \alpha_l (1 - \beta_1) \sqrt{\frac{1-\beta_2}{(V'_1)^2 d}} \| \nabla_{\mG_b} F(x_t) \|^2 
+ \eta \alpha_l (1 - \beta_1) \sqrt{\frac{1-\beta_2}{(V'_1)^2 d}}  \ip{\nabla_{\mG_b} F(x_t)}{\nabla_{\mG_b} F(x_t) - \E[p_{t, \mG_b}]} \\ 
&\quad+ \eta \alpha_u (1 - \beta_1) \sum_{j \in \mG_b} \E |(\nabla_{\mG_b} F(x_t) - p_{t, \mG_b})_j|. 
\label{eq:3_up}
\end{align}

Combining (\ref{eq:var_up}), (\ref{eq:1_up}), (\ref{eq:2_up}) and (\ref{eq:3_up}) together, and taking expectation conditional on $x_t$, we have
\begin{align*}
&\E[ F(x_{t+1}) ] \\
&\leq F(x_t) + \sum_{b \in [B]} \E[ \mcir{1} ] + \sum_{b \in [B]} \frac{\eta^2 \alpha_u^2 L_b}{2} \\
&\leq F(x_t) + \sum_{b \in [B]} \E[ \mcir{2} ] + \sum_{b \in [B]} \E[ \mcir{3} ] + \sum_{b \in [B]} \frac{\eta^2 \alpha_u^2 L_b}{2} \\
&\leq F(x_t) + \sum_{b \in [B]} \frac{\eta^2 \alpha_u^2 L_b}{2} \\
&\quad + \sum_{b \in [B]} - \eta \alpha_l (1 - \beta_1) \sqrt{\frac{1-\beta_2}{(V'_1)^2 d}} \| \nabla_{\mG_b} F(x_t) \|^2 \\
& \quad + \eta \alpha_u (1 - \beta_1) \sum_{b \in [B]} \sum_{j \in \mG_b} \E |(\nabla_{\mG_b} F(x_t) - p_{t, \mG_b})_j| \\
&\quad + \sum_{b \in [B]} \eta \alpha_l (1 - \beta_1) \sqrt{\frac{1-\beta_2}{(V'_1)^2 d}}  \ip{\nabla_{\mG_b} F(x_t)}{\nabla_{\mG_b} F(x_t) - \E[p_{t, \mG_b}]} \\
& \quad + \sum_{b \in [B]} \eta \beta_1 \alpha_u \sum_{j \in \mG_b} \E| (\nabla_{\mG_b} F(x_{t}) - \tm_{t, \mG_b})_j| \\
&= F(x_t) + \frac{\eta^2 \alpha_u^2 \|L\|_1}{2} \\
&\quad - \eta \alpha_l (1 - \beta_1) \sqrt{\frac{1-\beta_2}{(V'_1)^2 d}} \| \nabla F(x_t) \|^2 
\\
& \quad + \eta \alpha_u (1 - \beta_1) \sum_{b \in [B]} \sum_{j \in \mG_b} \E |(\nabla_{\mG_b} F(x_t) - p_{t, \mG_b})_j| \\ 
&\quad + \eta \alpha_l (1 - \beta_1) \sqrt{\frac{1-\beta_2}{(V'_1)^2 d}}  \ip{\nabla F(x_t)}{\nabla F(x_t) - \E[p_{t}]} \\
& \quad + \eta \beta_1 \alpha_u \sum_{b \in [B]} \sum_{j \in \mG_b} \E| (\nabla_{\mG_b} F(x_{t}) - \tm_{t, \mG_b})_j| \\
&\leq F(x_t) + \frac{\eta^2 \alpha_u^2 \|L\|_1}{2} 
- \eta \alpha_l (1 - \beta_1) \sqrt{\frac{1-\beta_2}{(V'_1)^2 d}} \| \nabla F(x_t) \|^2 \\
&\quad+ \eta \alpha_u (1 - \beta_1) V_2 \comment{$\sum_{b \in [B]} \sum_{j \in \mG_b} \E |(\nabla_{\mG_b} F(x_t) - p_{t, \mG_b})_j| \leq V_2$} \\ 
&\quad + \eta \alpha_l (1 - \beta_1) \sqrt{\frac{(1-\beta_2)G^2}{(V'_1)^2 }} V_3 \comment{$\|\nabla F(x_t) - \E[p_{t}]\| \leq V_3$} \\
&\quad+ \eta \beta_1 \alpha_u \sum_{b \in [B]} \sum_{j \in \mG_b} \E| (\nabla_{\mG_b} F(x_{t}) - \tm_{t, \mG_b})_j| \\
\end{align*}
Telescoping, taking total expectation, and dividing by $T$, we have
\begin{align*}
&\frac{\E[ F(x_{T+1}) - F(x_1) ]}{T} \\
&\leq \frac{\eta^2 \alpha_u^2 \|L\|_1}{2} 
- \eta \alpha_l (1 - \beta_1) \sqrt{\frac{1-\beta_2}{(V'_1)^2 d}} \sum_{t=1}^{T} \frac{\E \| \nabla F(x_t) \|^2}{T}  
+ \eta \alpha_u (1 - \beta_1) V_2  
+ \eta \alpha_l (1 - \beta_1) \sqrt{\frac{(1-\beta_2)G^2}{(V'_1)^2 }} V_3  \\
&\quad + \eta \beta_1 \alpha_u \sum_{t=1}^T \frac{1}{T} \sum_{b \in [B]} \sum_{j \in \mG_b} \E| (\nabla_{\mG_b} F(x_{t}) - \tm_{t, \mG_b})_j| \\
&\leq \frac{\eta^2 \alpha_u^2 \|L\|_1}{2} 
- \eta \alpha_l (1 - \beta_1) \sqrt{\frac{1-\beta_2}{(V'_1)^2 d}} \sum_{t=1}^{T} \frac{\E \| \nabla F(x_t) \|^2}{T} 
+ \eta \alpha_u (1 - \beta_1) V_2  
+ \eta \alpha_l (1 - \beta_1) \sqrt{\frac{(1-\beta_2)G^2}{(V'_1)^2 }} V_3 \\ 
&\quad + \eta \beta_1 \alpha_u \sum_{t=1}^T \frac{1}{T} \sum_{b \in [B]} \sum_{j \in \mG_b} \left[ \frac{\beta_1}{1-\beta_1} L_{b,j} \eta \alpha_u
+ \E| (p_{t, \mG_b} - \nabla_{\mG_b} F(x_{t}))_j |
+ \sum_{\tau = 1}^{t-1}  \frac{\beta_1^\tau}{1-\beta_1^t} \E| (p_{t-\tau, \mG_b} - \nabla_{\mG_b} F(x_{t-\tau}))_j | \right] \comment{Lemma~\ref{lem:mom_err}} \\
&\leq \frac{\eta^2 \alpha_u^2 \|L\|_1}{2} 
- \eta \alpha_l (1 - \beta_1) \sqrt{\frac{1-\beta_2}{(V'_1)^2 d}} \sum_{t=1}^{T} \frac{\E \| \nabla F(x_t) \|^2}{T} 
+ \eta \alpha_u (1 - \beta_1) V_2  
+ \eta \alpha_l (1 - \beta_1) \sqrt{\frac{(1-\beta_2)G^2}{(V'_1)^2 }} V_3 \\ 
&\quad + \eta \beta_1 \alpha_u \left[ \frac{\beta_1}{1-\beta_1} \|L\|_1 \eta \alpha_u
+ V_2
+ V_2 \sum_{\tau = 1}^{t-1}  \frac{\beta_1^\tau}{1-\beta_1}  \right] \\
&\leq \frac{\eta^2 \alpha_u^2 \|L\|_1}{2} 
- \eta \alpha_l (1 - \beta_1) \sqrt{\frac{1-\beta_2}{(V'_1)^2 d}} \sum_{t=1}^{T} \frac{\E \| \nabla F(x_t) \|^2}{T} 
+ \eta \alpha_u (1 - \beta_1) V_2  
+ \eta \alpha_l (1 - \beta_1) \sqrt{\frac{(1-\beta_2)G^2}{(V'_1)^2 }} V_3 \\ 
&\quad + \eta \beta_1 \alpha_u \left[ \frac{\beta_1}{1-\beta_1} \|L\|_1 \eta \alpha_u
+ \frac{(2 - 2\beta_1 + \beta_1^2) V_2}{(1-\beta_1)^2}  \right].
\comment{$\sum_{\tau = 1}^{t-1}  \frac{\beta_1^\tau}{1-\beta_1} \leq \frac{\beta_1}{(1-\beta_1)^2}$}
\end{align*}

By re-arranging the terms and dividing by $\eta \alpha_l (1 - \beta_1) \sqrt{\frac{1-\beta_2}{(V'_1)^2 d}}$, we have
\begin{align*}
&\frac{\sum_{t=1}^{T} \E \| \nabla F(x_t) \|^2}{T} \\
&\leq \frac{\sqrt{d} V'_1 \E[ F(x_{1}) - F(x_{T+1}) ]}{T \eta \alpha_l (1-\beta_1) \sqrt{1-\beta_2}}
+ \frac{\eta \sqrt{d} V'_1 \alpha_u^2 (1-\beta_1 + 2\beta_1^2) \|L\|_1 }{2 \sqrt{1-\beta_2} (1-\beta_1)^2 \alpha_l} 
+ \frac{\sqrt{d} V'_1 \alpha_u (1-\beta_1^2+\beta_1^3)}{\sqrt{1-\beta_2} (1-\beta_1)^2 \alpha_l} V_2 
+ \sqrt{d} G V_3 \\
&\leq \frac{\sqrt{d} V'_1 \E[ F(x_{1}) - F(x_{*}) ]}{T \eta \alpha_l (1-\beta_1) \sqrt{1-\beta_2}}
+ \frac{\eta \sqrt{d} V'_1 \alpha_u^2 (1-\beta_1 + 2\beta_1^2) \|L\|_1 }{2 \sqrt{1-\beta_2} (1-\beta_1)^2 \alpha_l} 
+ \frac{\sqrt{d} V'_1 \alpha_u (1-\beta_1^2+\beta_1^3)}{\sqrt{1-\beta_2} (1-\beta_1)^2 \alpha_l} V_2 
+ \sqrt{d} G V_3.
\end{align*}

\end{proof}

\begin{corollary}
For CLAN with full precision~(i.e., LANS), we have $p_t = \frac{1}{n} \sum_{i=1}^ng_{t, i}$, $V'_1 = G' = G + \epsilon$, $V_2 = \frac{\|\sigma\|_1}{\sqrt{ns}}$, $V_3 = 0$, where $\E[ (\nabla F(x) - \nabla f(x))_j^2 ] \leq \sigma_j^2, \forall x \in \R^d, j \in [d]$, and $\|\sigma\|_1 = \sum_{j \in [d]} \sigma_j$.
Thus, we have the following error bound: 
\begin{align*}
\frac{\sum_{t=1}^{T} \E \| \nabla F(x_t) \|^2}{T} 
\leq \frac{\sqrt{d} G' \E[ F(x_{1}) - F(x_{*}) ]}{T \eta \alpha_l (1-\beta_1) \sqrt{1-\beta_2}}
+ \frac{\eta \sqrt{d} G' \alpha_u^2 (1-\beta_1 + 2\beta_1^2) \|L\|_1 }{2 \sqrt{1-\beta_2} (1-\beta_1)^2 \alpha_l} 
+ \frac{\sqrt{d} G' \alpha_u [(1-\beta_1)^2+\beta_1]}{\sqrt{1-\beta_2} (1-\beta_1)^2 \alpha_l} \frac{\|\sigma\|_1}{\sqrt{ns}}.
\end{align*}
Furthermore, taking $\eta = \frac{1}{\sqrt{T}}$, $n s \propto T$, we have $\frac{\sum_{t=1}^{T} \E \| \nabla F(x_t) \|^2}{T} 
\leq \mathcal{O} \left(\frac{1}{\sqrt{T}}\right)$.
\end{corollary}

\begin{corollary}
For CLAN with unbiased compressors with the approximation $\E[\|\mathcal{C}(v)\|^2] \leq \omega \|v\|^2$, we have $p_t = \mathcal{C}\left( \frac{1}{n} \sum_{i \in [n]} \mathcal{C}(g_{t, i}) \right)$, where $\E[p_t] = \frac{1}{n}\sum_{i \in [n]} g_{t, i}$. 
Thus, we have $V'_1 = \left[1 + d \sqrt{ \omega-1 + \frac{\omega (\omega-1)}{n} }\right] G + \epsilon$, $V_2 = \frac{\|\sigma\|_1}{\sqrt{n s}} + dG \sqrt{ \omega-1 + \frac{\omega (\omega-1)}{n} }$, $V_3 = 0$, where $\E[ (\nabla F(x) - \nabla f(x))_j^2 ] \leq \sigma_j^2, \forall x \in \R^d, j \in [d]$, and $\|\sigma\|_1 = \sum_{j \in [d]} \sigma_j$.
Thus, we have the following error bound: 
\begin{align*}
& \frac{\sum_{t=1}^{T} \E \| \nabla F(x_t) \|^2}{T} \\
& \leq \frac{\sqrt{d} V'_1 \E[ F(x_{1}) - F(x_{*}) ]}{T \eta \alpha_l (1-\beta_1) \sqrt{1-\beta_2}}
+ \frac{\eta \sqrt{d} V'_1 \alpha_u^2 (1-\beta_1 + 2\beta_1^2) \|L\|_1 }{2 \sqrt{1-\beta_2} (1-\beta_1)^2 \alpha_l} 
+ \frac{\sqrt{d} V'_1 \alpha_u [(1-\beta_1)^2+\beta_1]}{\sqrt{1-\beta_2} (1-\beta_1)^2 \alpha_l} \left( \frac{\|\sigma\|_1}{\sqrt{n s}} + dG \sqrt{ \omega-1 + \frac{\omega (\omega-1)}{n} } \right).
\end{align*}
Furthermore, taking $\eta = \frac{1}{\sqrt{T}}$, $n s \propto T$, and using compressors with $\omega \leq 1+\frac{1}{T}$, we have $\frac{\sum_{t=1}^{T} \E \| \nabla F(x_t) \|^2}{T} 
\leq \mathcal{O} \left(\frac{1}{\sqrt{T}}\right)$.
\end{corollary}

\begin{proof}

We first establish the following upper bound:
\begin{align*}
&  \E\left\|\frac{1}{n} \sum_{i \in [n]} g_{t, i} - p_{t}\right\|^2 \\
&= - \left\|\frac{1}{n} \sum_{i \in [n]} g_{t, i} \right\|^2 + \E\left\|p_{t}\right\|^2 \\
&\leq - \left\|\frac{1}{n} \sum_{i \in [n]} g_{t, i} \right\|^2 + \omega \E\left\| \frac{1}{n} \sum_{i \in [n]} \mathcal{C}(g_{t, i}) \right\|^2 \\
&= - \frac{1}{n^2} \left\| \sum_{i \in [n]} g_{t, i} \right\|^2 + \omega \frac{1}{n^2} \E\left\| \sum_{i \in [n]} \mathcal{C}(g_{t, i}) \right\|^2 \\
&= - \frac{1}{n^2} \left\| \sum_{i \in [n]} g_{t, i} \right\|^2 + \omega \frac{1}{n^2} \left[ \sum_{i \neq j} \ip{g_{t, i}}{g_{t, j}} + \sum_{i \in [n]} \E\|\mathcal{C}(g_{t, i})\|^2 \right] \\
&\leq - \frac{1}{n^2} \left\| \sum_{i \in [n]} g_{t, i} \right\|^2 + \omega \frac{1}{n^2} \left[ \sum_{i \neq j} \ip{g_{t, i}}{g_{t, j}} + \sum_{i \in [n]} \omega \|g_{t, i}\|^2 \right] \\
&= - \frac{1}{n^2} \left\| \sum_{i \in [n]} g_{t, i} \right\|^2 + \omega \frac{1}{n^2} \left[ \left\| \sum_{i \in [n]} g_{t, i} \right\|^2 + \sum_{i \in [n]} (\omega-1) \|g_{t, i}\|^2 \right] \\
&= (\omega - 1) \left\| \frac{1}{n} \sum_{i \in [n]} g_{t, i} \right\|^2 + \frac{\omega (\omega-1)}{n^2} \sum_{i \in [n]}  \|g_{t, i}\|^2  \\
&\leq \left[ \omega-1 + \frac{\omega (\omega-1)}{n}  \right] d G^2.
\end{align*}

Thus, we have
\begin{align*}
&\sum_{j \in [d]} \E \left|\left( p_t - \frac{1}{n} \sum_{i \in [n]} g_{t, i} \right)_j \right| \\
&\leq \sqrt{d} \E \left\| p_t - \frac{1}{n} \sum_{i \in [n]} g_{t, i} \right\| \\
&\leq \sqrt{d} \sqrt{ \E \left\| p_t - \frac{1}{n} \sum_{i \in [n]} g_{t, i} \right\|^2 } \\
&\leq dG \sqrt{ \omega-1 + \frac{\omega (\omega-1)}{n} }.
\end{align*}

Second, we establish the upper bound $(p_t)_j \leq V_1, \forall j \in [d]$.
\begin{align*}
&\E(p_t)_j \\
&\leq \left( \frac{1}{n} \sum_{i \in [n]} g_{t, i} \right)_j + \E\left|\left( p_t - \frac{1}{n} \sum_{i \in [n]} g_{t, i} \right)_j \right| \\
&\leq G + \E\left|\left( p_t - \frac{1}{n} \sum_{i \in [n]} g_{t, i} \right)_j \right| \\
&\leq G + \sum_{j \in [d]} \E\left|\left( p_t - \frac{1}{n} \sum_{i \in [n]} g_{t, i} \right)_j \right| \\
&\leq \left[1 + d \sqrt{ \omega-1 + \frac{\omega (\omega-1)}{n} }\right] G = V_1.
\end{align*}

Then, we establish the upper bound $\sum_{j \in [d]} \E |(\nabla F(x_t) - p_{t})_j| \leq V_2$.
\begin{align*}
&\sum_{j \in [d]} \E |(\nabla F(x_t) - p_{t})_j| \\
&= \sum_{j \in [d]} \E |(\nabla F(x_t) - \frac{1}{n} \sum_{i \in [n]} g_{t, i} + \frac{1}{n} \sum_{i \in [n]} g_{t, i} - p_{t})_j| \\
&\leq \frac{\|\sigma\|_1}{\sqrt{n s}} + \sum_{j \in [d]} \E |( \frac{1}{n} \sum_{i \in [n]} g_{t, i} - p_{t})_j| \\
&\leq \frac{\|\sigma\|_1}{\sqrt{n s}} + dG \sqrt{ \omega-1 + \frac{\omega (\omega-1)}{n} } = V_2.
\end{align*}

Finally, unbiased compressors imply $\E[p_{t}] = \nabla F(x_t)$.
Thus, $\|\nabla F(x_t) - \E[p_{t}]\| = 0$.

\end{proof}


\begin{lemma}\label{lem:ef_bound}
For CLAN with biased compressors and error feedback, we have
\begin{align*}
\left\| \tilde{e}_t + \frac{1}{n} \sum_{i \in [n]} e_{t, i} \right\|
\leq \left[ \frac{\sqrt{d(1-\delta)}}{1 - \sqrt{1-\delta}} + \frac{2 \sqrt{1-\delta}}{1-\sqrt{1-\delta}} \left(1+\frac{\sqrt{d(1-\delta)}}{1 - \sqrt{1-\delta}}\right) \right] G.
\end{align*}
\end{lemma}

\begin{proof}

First, we establish the upper bound of $\|e_{t, i}\|$, using $e_{0, i} = 0$.
\begin{align*}
&\|e_{t+1, i}\| \\
&= \| g_{t,i} + e_{t,i} - \mathcal{C}(g_{t,i} + e_{t,i}) \| \\
&\leq \sqrt{1-\delta} \|g_{t,i} + e_{t,i}\| \\
&\leq \sqrt{1-\delta} \|g_{t,i} \| + \sqrt{1-\delta} \|e_{t,i}\| \\
&\leq \sqrt{d(1-\delta)} G + \sqrt{1-\delta} \|e_{t,i}\| \\
&\leq \sum_{k = 0}^t (\sqrt{1-\delta})^k \sqrt{d(1-\delta)} G \\
&\leq \frac{\sqrt{d(1-\delta)}}{1 - \sqrt{1-\delta}} G.
\end{align*}

Then, we establish the upper bound of $\|\tilde{e}_t\|$, using $\tilde{e}_0 = 0$.
\begin{align*}
&\|\tilde{e}_{t+1}\| \\
&\leq \sqrt{1-\delta} \left\| \tilde{e}_{t} + \frac{1}{n} \sum_{i \in [n]} \mathcal{C}( g_{t, i} + e_{t, i} ) \right\| \\
&\leq \sqrt{1-\delta} \| \tilde{e}_{t} \| + \sqrt{1-\delta} \frac{1}{n} \sum_{i \in [n]} \|  \mathcal{C}( g_{t, i} + e_{t, i} ) - (g_{t, i} + e_{t, i}) + (g_{t, i} + e_{t, i}) \| \\
&\leq \sqrt{1-\delta} \| \tilde{e}_{t} \| + \sqrt{1-\delta} \frac{1}{n} \sum_{i \in [n]} \|  \mathcal{C}( g_{t, i} + e_{t, i} ) - (g_{t, i} + e_{t, i}) \| + \sqrt{1-\delta} \frac{1}{n} \sum_{i \in [n]} \| g_{t, i} + e_{t, i} \| \\
&\leq \sqrt{1-\delta} \| \tilde{e}_{t} \| + (1-\delta) \frac{1}{n} \sum_{i \in [n]} \| g_{t, i} + e_{t, i} \| + \sqrt{1-\delta} \frac{1}{n} \sum_{i \in [n]} \| g_{t, i} + e_{t, i} \| \\
&\leq \sqrt{1-\delta} \| \tilde{e}_{t} \| + 2 \sqrt{1-\delta} \frac{1}{n} \sum_{i \in [n]} [ \| g_{t, i} \| + \| e_{t, i} \| ] \\
&\leq \sqrt{1-\delta} \| \tilde{e}_{t} \| + 2 \sqrt{1-\delta} \left(1+\frac{\sqrt{d(1-\delta)}}{1 - \sqrt{1-\delta}}\right) G \\
&\leq 2 \sqrt{1-\delta} \left(1+\frac{\sqrt{d(1-\delta)}}{1 - \sqrt{1-\delta}}\right) G \sum_{k = 0}^t (\sqrt{1-\delta})^k \\
&\leq 2 \frac{\sqrt{1-\delta}}{1-\sqrt{1-\delta}} \left(1+\frac{\sqrt{d(1-\delta)}}{1 - \sqrt{1-\delta}}\right) G.
\end{align*}

Thus, we have
\begin{align*}
\left\| \tilde{e}_t + \frac{1}{n} \sum_{i \in [n]} e_{t, i} \right\|
\leq | \tilde{e}_t \| + \frac{1}{n} \sum_{i \in [n]} \| e_{t, i} \|
\leq \left[ \frac{\sqrt{d(1-\delta)}}{1 - \sqrt{1-\delta}} + \frac{2 \sqrt{1-\delta}}{1-\sqrt{1-\delta}} \left(1+\frac{\sqrt{d(1-\delta)}}{1 - \sqrt{1-\delta}}\right) \right] G.
\end{align*}

\end{proof}

\begin{corollary}
For CLAN with biased compressors and error feedback, we have $p_t =  \mathcal{C}\left( \frac{1}{n} \sum_{i \in [n]} \mathcal{C}(g_{t, i} + e_{t, i}) + \tilde{e}_t \right)$.
Thus, we have
\begin{align*}
V_1 &= G + \sqrt{d} V_3, \\
V_2 &=
\frac{\|\sigma\|_1}{\sqrt{n s}} + \sqrt{d} V_3, \\
V_3 &=
2 \left[ \frac{\sqrt{d(1-\delta)}}{1 - \sqrt{1-\delta}} + \frac{2 \sqrt{1-\delta}}{1-\sqrt{1-\delta}} \left(1+\frac{\sqrt{d(1-\delta)}}{1 - \sqrt{1-\delta}}\right) \right] G,
\end{align*}
which results in the following error bound: 
\begin{align*}
&\frac{\sum_{t=1}^{T} \E \| \nabla F(x_t) \|^2}{T} \\
& \leq \frac{\sqrt{d} (G + \sqrt{d} V_3 + \epsilon) \E[ F(x_{1}) - F(x_{*}) ]}{T \eta \alpha_l (1-\beta_1) \sqrt{1-\beta_2}}
+ \frac{\eta \sqrt{d} (G + \sqrt{d} V_3 + \epsilon) \alpha_u^2 (1-\beta_1 + 2\beta_1^2) \|L\|_1 }{2 \sqrt{1-\beta_2} (1-\beta_1)^2 \alpha_l} \\
&\quad+ \frac{\sqrt{d} (G + \sqrt{d} V_3 + \epsilon) \alpha_u [(1-\beta_1)^2+\beta_1]}{\sqrt{1-\beta_2} (1-\beta_1)^2 \alpha_l} \left( \frac{\|\sigma\|_1}{\sqrt{n s}} + \sqrt{d} V_3 \right) 
+ \sqrt{d} G V_3
\end{align*}
Furthermore, taking $\eta = \frac{1}{\sqrt{T}}$, $n s \propto T$, and using compressors with $\delta \geq 1 - \frac{1}{(\sqrt{T}-1)^2}$, we have $\frac{\sum_{t=1}^{T} \E \| \nabla F(x_t) \|^2}{T} 
\leq \mathcal{O} \left(\frac{1}{\sqrt{T}}\right)$.
\end{corollary}

\begin{proof}

We first establish the following upper bound:
\begin{align}
& \sum_{j \in [d]} |( \frac{1}{n} \sum_{i \in [n]} g_{t, i} - p_{t})_j| \nonumber\\
&\leq \sqrt{d} \left\| \frac{1}{n} \sum_{i \in [n]} g_{t, i} - p_t \right\| \label{eq:gpt}\\
&= \sqrt{d} \left\| \frac{1}{n} \sum_{i \in [n]} g_{t, i} -  \left[\frac{1}{n} \sum_{i \in [n]} \mathcal{C}(g_{t, i} + e_{t, i}) + \tilde{e}_t - \tilde{e}_{t+1}\right]\right\| \nonumber\\
&= \sqrt{d} \left\| \frac{1}{n} \sum_{i \in [n]} g_{t, i} -  \left[\frac{1}{n} \sum_{i \in [n]} \left[(g_{t, i} + e_{t, i}) - e_{t+1, i}\right] + \tilde{e}_t - \tilde{e}_{t+1}\right]\right\| \nonumber\\
&= \sqrt{d} \left\| \tilde{e}_t - \tilde{e}_{t+1} + \frac{1}{n} \sum_{i \in [n]} [ e_{t,i} - e_{t+1,i} ] \right\| \nonumber\\
&= \sqrt{d} \left\| \frac{1}{n} \sum_{i \in [n]} e_{t+1,i} + \tilde{e}_{t+1} \right\| + \sqrt{d} \left\| \frac{1}{n} \sum_{i \in [n]} e_{t,i} + \tilde{e}_{t} \right\| \nonumber\\
&= 2 \sqrt{d} \left[ \frac{\sqrt{d(1-\delta)}}{1 - \sqrt{1-\delta}} + \frac{2 \sqrt{1-\delta}}{1-\sqrt{1-\delta}} \left(1+\frac{\sqrt{d(1-\delta)}}{1 - \sqrt{1-\delta}}\right) \right] G. \nonumber \comment{Lemma~\ref{lem:ef_bound}}
\end{align}

Second, we establish the upper bound $(p_t)_j \leq V_1, \forall j \in [d]$.
\begin{align*}
&(p_t)_j \\
&\leq \left( \frac{1}{n} \sum_{i \in [n]} g_{t, i} \right)_j + \left|\left( p_t - \frac{1}{n} \sum_{i \in [n]} g_{t, i} \right)_j \right| \\
&\leq G + \left|\left( p_t - \frac{1}{n} \sum_{i \in [n]} g_{t, i} \right)_j \right| \\
&\leq G + \sum_{j \in [d]} \left|\left( p_t - \frac{1}{n} \sum_{i \in [n]} g_{t, i} \right)_j \right| \\
&\leq \left[ 1 + 2 \sqrt{d} \left( \frac{\sqrt{d(1-\delta)}}{1 - \sqrt{1-\delta}} + \frac{2 \sqrt{1-\delta}}{1-\sqrt{1-\delta}} \left(1+\frac{\sqrt{d(1-\delta)}}{1 - \sqrt{1-\delta}}\right) \right) \right] G  = V_1.
\end{align*}

Then, we establish the upper bound $\sum_{j \in [d]} \E |(\nabla F(x_t) - p_{t})_j| \leq V_2$.
\begin{align*}
&\sum_{j \in [d]} \E |(\nabla F(x_t) - p_{t})_j| \\
&= \sum_{j \in [d]} \E |(\nabla F(x_t) - \frac{1}{n} \sum_{i \in [n]} g_{t, i} + \frac{1}{n} \sum_{i \in [n]} g_{t, i} - p_{t})_j| \\
&\leq \frac{\|\sigma\|_1}{\sqrt{n s}} + \sum_{j \in [d]} \E |( \frac{1}{n} \sum_{i \in [n]} g_{t, i} - p_{t})_j| \\
&\leq \frac{\|\sigma\|_1}{\sqrt{n s}} + 2 \sqrt{d} \left[ \frac{\sqrt{d(1-\delta)}}{1 - \sqrt{1-\delta}} + \frac{2 \sqrt{1-\delta}}{1-\sqrt{1-\delta}} \left(1+\frac{\sqrt{d(1-\delta)}}{1 - \sqrt{1-\delta}}\right) \right] G = V_2.
\end{align*}

Finally, for $V_3$, we have
\begin{align*}
&\|\nabla F(x_t) - \E[p_{t}]\| \\
&= \left \|\nabla F(x_t) - \E\left [p_{t} - \frac{1}{n} \sum_{i \in [n]} g_{t, i} + \frac{1}{n} \sum_{i \in [n]} g_{t, i} \right] \right\| \\
&= \left\|\E \left[ p_{t} - \frac{1}{n} \sum_{i \in [n]} g_{t, i} \right] \right\| \\
&\leq \E\left\| \left[ p_{t} - \frac{1}{n} \sum_{i \in [n]} g_{t, i} \right] \right\| \\
&\leq 2 \left[ \frac{\sqrt{d(1-\delta)}}{1 - \sqrt{1-\delta}} + \frac{2 \sqrt{1-\delta}}{1-\sqrt{1-\delta}} \left(1+\frac{\sqrt{d(1-\delta)}}{1 - \sqrt{1-\delta}}\right) \right] G = V_3,
\end{align*}
where we use Jensen's inequality in the first inequality. The last inequality is due to (\ref{eq:gpt}).

\end{proof}

\end{document}